\documentclass[a4paper,10pt]{article}

\usepackage{amsthm,amsmath,amssymb,amscd,verbatim,graphicx,setspace,soul,hyperref}
\usepackage[dvipsnames]{xcolor}
\usepackage{centernot} 
\usepackage{bm}  
\usepackage{enumitem}
\setlist[enumerate]{itemsep=0mm}
\usepackage{mathtools}
\usepackage{soul} 
\usepackage{tikz-cd}

\usepackage[utf8]{inputenc}
\usepackage[english]{babel}
 
\usepackage{geometry}
\usepackage{marginnote}

\usepackage[round]{natbib}
\let\originalleft\left
\let\originalright\right
\renewcommand{\left}{\mathopen{}\mathclose\bgroup\originalleft}
\renewcommand{\right}{\aftergroup\egroup\originalright}

\newcommand{\Set}[2]{%
  \{\, #1 \mid #2 \, \}%
}

\def\AA{\mathcal{A}}  
  
  \def\II{\mathcal{I}}
\def\JJ{\mathcal{J}}

\def\III{\mathbb{I}} \def\JJJ{\mathbb{J}} \def\KKK{\mathbb{K}} 
  \def\NNN{\mathbb{N}} 
   
\def\RRR{\mathbb{R}}   
 \def\VVV{\mathbb{V}}

\def\<{\left<} \def\>{\right>} 
\def\({\left(} \def\){\right)}

\DeclareMathOperator\id{id}

\DeclareMathOperator\Span{Span}

\theoremstyle{plain}
\newtheorem{thm}{Theorem}
\newtheorem{prop}[thm]{Proposition}

\newtheorem{lemma}[thm]{Lemma}

\theoremstyle{definition}

\newtheorem{example}[thm]{Example}
\newtheorem*{original standing assumptions}{Original Harsanyi Assumptions}
\newtheorem*{standing assumptions}{Standing Assumptions}
\newtheorem{rem}[thm]{Remark}

\newcommand{\ProdI}{\prod_{i \in \III}}
\newcommand{\SumI}{\bigoplus_{i \in \III}}

\newcommand\ov[1]{\overline{#1}}
\newcommand\fI{f_\III}
\newcommand\ovfI{\ov{f}_\III}
\newcommand\gI{g_\III}

\newcommand\uI{u_\III}
\newcommand\UI{U_\III}

\def\inv{^{-1}}

\def\Geq{\succsim}
\def\Leq{\precsim}
\def\sGeq{\succ}

\def\Eq{\sim}

\def\nLeq{\not\Leq}
\def\Incom{\curlywedge}

\def\GeqV{\succsim_V}
\def\LeqV{\precsim_V}
\def\sGeqV{\succ_{V}}
\def\EqV{\sim_V}

\def\IncomV{\Incom_V}

\def\GeqVi{\succsim_{V_i}}

\def\sGeqVi{\succ_{V_i}}

\def\IncomVi{\Incom_{V_i}}

\def\GeqVj{\succsim_{V_j}}

\def\GeqP{\succsim_\mathrm{P}}     

\def\sGeqP{\succ_\mathrm{P}}

\def\IncomP{\curlywedge_\mathrm{P}}

\def\aa{\alpha}

\def\ll{\lambda}

\newcommand{\restr}[1]{\lower.45ex\hbox{$|$}\lower.5ex\hbox{}_{#1}} 

\def\IFF{\Leftrightarrow}

\title{Aggregation for potentially infinite populations without continuity or completeness%
\thanks{David McCarthy thanks the Research Grants Council of the Hong Kong Special Administrative Region, China (HKU 750012H) for support. Teruji Thomas thanks the Leverhulme trust for funding through the project `Population Ethics: Theory and Practice' (RPG-2014-064). An earlier version of this paper appeared as `Aggregation for general populations without continuity or completeness' MPRA Paper No. 80820 (2017).}
}
\author{David McCarthy%
\footnote{Corresponding author, Dept. of Philosophy, University of Hong Kong, Hong Kong, 
 \texttt{davidmccarthy1@gmail.com}} 
\and 
Kalle Mikkola%
\footnote{Dept. of Mathematics and Systems Analysis, Aalto University, Finland,
\texttt{kalle.mikkola@iki.fi}} 
\and 
Teruji Thomas%
\footnote{Global Priorities Institute, University of Oxford, United Kingdom, 
\texttt{teru.thomas@oxon.org}}
}

\def\V{V}
\def\X{X}

\def\THEN{\implies}

\newcommand{\CP}{C_P}

\begin{document}
\date{\vspace{-5ex}}

\maketitle

\begin{abstract}
We present an abstract social aggregation theorem. 
Society, and each individual, has a preorder that may be interpreted as expressing values or beliefs. 
The preorders are allowed to violate both completeness and continuity,  and the population is allowed to be infinite. 
The preorders are only assumed to 
be represented by  functions with values in partially ordered vector spaces, and whose product has convex range.   
This includes all preorders that satisfy strong independence. 
Any Pareto indifferent social preorder is then shown to be  represented by a linear transformation 
of the representations of the  individual preorders. 
Further Pareto conditions on the social  preorder correspond to positivity conditions on the transformation.
When all the Pareto conditions hold and the population is finite, the  social preorder is represented by a sum of individual preorder representations. 
We provide two applications. 
The first yields an extremely general version of Harsanyi's social aggregation theorem. 
The second  generalizes a classic result about linear  opinion  pooling.


\end{abstract}

\section{Introduction}

Individuals may have incomplete 
values and beliefs.%
\footnote{The literature on the topics of this paragraph is vast. References that are directly relevant to our approach are given in section~\ref{S:povs}.} 
They may be undecided which of two goods is 
preferable, 
or which of two events is more likely. 
The former is reflected, for example, in the 
literature developing multi-utility theory, 
and the latter in 
work 
on representing beliefs 
by sets of probability measures, 
and on decision theories that feature such sets. 
Individuals may also have discontinuous values and beliefs.%
\footnote{The applications we discuss are either decision theoretic (values) or to do with comparative likelihood (beliefs). Each subject has its own family of continuity axioms, but we rarely need to specify these, and we rely on context to determine which family we are discussing.} 
For example, they may see some goods as being infinitely more valuable than others (or, equivalently, see some goods as only having significance as tiebreakers). 
In addition, allowing individuals to regard some events as  
infinitesimally likely 
provides a solution to a number of problems to do with conditional probability, decision theory, game theory and conditional preference.

A standard question is how to aggregate the 
values 
or  beliefs of individuals to form a collective view. 
Here we also wish to allow the population to be infinite. One rationale is that even in a finite society, 
a decision maker having incomplete information about the values and beliefs of members of society may wish to model each individual as an infinite set of types.%
\footnote{In the context of values, see e.g. \citet{lZ1997}, 
attributing the idea to \citet{jH1967}; for beliefs, see e.g. \citet{fH2015}.} 
But the most obvious rationale comes from the much discussed problem of intergenerational equity. Here, the 
possibility  
that society will extend indefinitely into the future 
requires an infinite number of people. It is also commonly modelled by an infinite sequence of generations, each one with a social utility function; in such a model, the generations play the role of `individuals' whose interests are to be aggregated.%
\footnote{\label{fn:infinite utility streams}The literature on this topic is largely shaped by the approaches of \citet{fR1928} (opposing impatience), and of \citet{tK1960} and \citet{pD1965} (requiring impatience). Our results are 
compatible with both approaches.}

Section~\ref{S:main} presents 
a family of abstract aggregation theorems. 
Each of possibly infinitely many individuals, and society, is equipped  with a preorder  on a given set.
We assume that these preorders are represented by  functions with values in partially ordered vector spaces, and whose product has convex range. 
The use of partially ordered vector spaces is explained and motivated in section \ref{S:povs};
axiomatizations of such representations are given in sections~\ref{S:harsanyi} and \ref{S:convexification}. 
Roughly speaking, our main result shows that Pareto indifference holds if and only if the social preorder can be represented by  a linear combination of the representations of the individual preorders.
Further Pareto conditions correspond to positivity conditions on the linear mapping.

Section~\ref{S:applications} illustrates the interest of the 
results with two  applications, corresponding respectively to the aggregation of values and of beliefs. The first yields an extremely general version of the celebrated social aggregation theorem of \citet{jH1955} that assumes only the central expected utility axiom of strong independence.  The other  shows that aggregate beliefs are given by linear pooling.  To reiterate, unlike other results in the literature, these results hold without any completeness or continuity assumptions, and  allow  for an infinite population.

We end in section \ref{S:literature} with a discussion of related literature, but for now we acknowledge \citet{pM1995} and \citet{DMM1995} for drawing attention to formal similarities between  preference  aggregation and  opinion  pooling, and emphasizing the  usefulness  of the convex range assumption we use below.

The proofs of our main theorems rely only on concepts from basic linear algebra, rather than any theorems from convex or functional analysis. All proofs are in the Appendix.

\section{Main results}\label{S:main}

\subsection{Representations in partially ordered vector spaces}\label{S:povs}
We are going to consider representations of preorders with values in partially ordered vector spaces. We first recall the basic definitions, and give some examples illustrating this type of representation.

Recall that a \emph{preorder} $\Geq$ is a binary relation that is reflexive and transitive;  we write $\Eq$ and $\sGeq$ for the symmetric and 
asymmetric parts of $\Geq$ respectively. Since preorders can be incomplete, 
we write $x \Incom y$ if neither $x \Geq y$ nor $y \Geq x$. 

Let $(X,\Geq)$ and $(X',\Geq')$ be preordered sets. A function $f \colon X \to X'$ is 
{\em increasing} if $x \Geq y \Rightarrow f(x) \Geq' f(y)$; {\em strictly increasing} if $x \Geq y \Rightarrow f(x) \Geq' f(y)$ and $x \sGeq y \Rightarrow f(x) \sGeq' f(y)$; 
a {\em representation} of $\Geq$ if $x \Geq y \Leftrightarrow f(x) \Geq' f(y)$; an {\em order embedding} if it is an injective representation; and an {\em order isomorphism} if it is a bijective representation.

A {\em preordered vector space} is a real
vector space $V$ with a 
 (possibly incomplete)  preorder $\Geq_V$ that is \emph{linear} in the sense that $v\Geq_V v' \Leftrightarrow \lambda v+w\Geq_V \lambda v'+w$, for all $v, v', w\in\VVV$ and $\lambda>0$.\footnote{A linear preorder, in our sense, is sometimes called a vector preorder.
See section \ref{S:literature} for generalisation to 
vector spaces over 
ordered fields other than $\RRR$. 
 } 
Note that we allow vector spaces to have infinite dimension.
A {\em partially ordered} vector space is a preordered vector space in which the linear preorder is a partial order; that is, it is anti-symmetric.  An {\em ordered} vector space is a partially ordered vector space in which the partial order is an order; that is, it is complete. 
When $L\colon V\to V'$ is a linear map between partially ordered vector spaces, $L$ is increasing if and only if it is \emph{positive}, in the sense that $v\Geq_V 0\Rightarrow Lv\Geq_{V'} 0$; and $L$ is strictly increasing if and only if it is \emph{strictly positive}, meaning that $v\Geq_V 0\Rightarrow Lv\Geq_{V'} 0$ and $v\sGeq_V0\Rightarrow Lv\sGeq_{V'} 0$.

In this paper we will be exclusively concerned with representations with values in partially ordered vector spaces.%
\footnote{See however Remark \ref{rem:preorder} for comments relevant to merely preordered vector spaces.}
We provide an axiomatic basis for using such representations 
for preference relations 
in 
Lemma~\ref{L:SI}, 
and
for comparative likelihood relations in 
Lemma~\ref{L:SI2} below, but for now we focus on examples.
The set $\RRR$ of real numbers, with the usual ordering, is a 
simple example of a partially ordered vector space, so our representations include familiar real-valued ones. Allowing for arbitrary partially ordered vector spaces allows for natural representations of incomplete and discontinuous preorders, as the following examples illustrate.

\begin{example}\label{ex:fame} 
Consider bundles of three goods (say fame, love, and money) represented by points in  $X=\RRR^3_+$. 
Let $V =  \RRR^3$, with the following 
linear 
partial order:
\[(x_1,x_2,x_3)\Geq_V (y_1,y_2,y_3)\iff (x_1\geq y_1,\,x_2\geq y_2, \text{ and, if } x_1=y_1\text{ and } x_2=y_2, \text{ then } x_3\geq y_3).
\]
Let $f \colon \RRR^3_+ \to V$ be the function 
$f(x_1,x_2,x_3)= (\tfrac23x_1+\tfrac13x_2,\tfrac13x_1+\tfrac 23x_2,x_3)$.  
Suppose person
$A$'s preference relation $\Geq_A$ 
over simple lotteries  over $\RRR^3_+$
is represented by expectations of $f$, 
thereby satisfying the expected utility axiom of strong independence.%
\footnote{\label{fn:SI}A preorder $\Geq$ on a convex set $X$ satisfies strong independence if for each $\aa \in (0,1)$, $x, y, z \in X$, $x \Geq y$ if and only if $\aa x + (1-\aa)z \Geq \aa y + (1-\aa) z$.} 
It violates the completeness axiom, because, for example, $(1,0,0) \Incom_A (0,1,0)$, reflecting the fact that $A$ 
finds fame and love only roughly comparable. 
(Note that the incomparability is limited in the sense that two units of fame are preferred to one unit of love, and vice versa.) The preferences violate standard decision theoretic notions of continuity, such as the Archimedean axiom of \citet{BG1954} and the mixture continuity axiom of \citet{HM1953}, 
reflecting the fact that while $A$ sees money as valuable, she finds fame (likewise love) infinitely more valuable. 
\end{example}

\begin{example}\label{ex:sphere}
Consider a sphere $S$ divided into open northern and southern hemispheres $S_+$ and $S_-$, and equator $S_0$. Let $\mu_+$, $\mu_-$, and $\mu_0$ be the uniform probability measures on $S_+$, $S_-$, and $S_0$, respectively. For every measurable set $A\subset S$ define 
$f(A)=(\mu_+(A\cap S_+),\mu_-(A\cap S_-),\mu_0(A\cap S_0))
\in V$, for $V=\RRR^3$ as in the previous example. 
This $f$ represents a likelihood preorder on the algebra $X$ 
of measurable subsets of $S$. The preorder is incomplete, since, for example, the hemispheres $S_+$ and $S_-$ are incomparable. 
(Here we have allowed the incomparability to be unlimited, in the sense that any positive-measure subsets of $S_+$ and $S_-$ are incomparable.)
Moreover, the equator $S_0$, though more likely than the null set, is less likely than the interior of any spherical triangle, no matter how small. Correspondingly, the likelihood preorder  violates 
standard continuity axioms for comparative likelihood,
such as the 
Archimedean axiom C6 of \citet{tF1973}, or the monotone continuity axiom of \citet{cV1964}  
and its weakening C8 in \citet{tF1973}. 
For both of these reasons, the likelihood preorder cannot 
be represented by a standard ($[0,1]$-valued) probability measure.  
\end{example}

In Example \ref{ex:fame}, the domain $X=\RRR^3_+$ is a convex set and the representation $f$ is \emph{mixture preserving}, 
i.e.  $f(\aa x+(1-\aa)x')=\aa f(x)+(1-\aa)f(x')$ for all $x,x'\in X,\aa\in [0,1]$.  
When $V=\RRR$, von Neumann-Morgenstern 
expected utility representations of preferences are the paradigmatic  representations of this type. 
In Example \ref{ex:sphere}, the domain $X$ is a Boolean algebra of sets, and the representation $f$ is a {\em vector measure}, meaning that it is finitely additive: $f(x \cup x') = f(x) + f(x')$ for disjoint $x,x'\in X$. When $V=\RRR$, probability measures are the paradigmatic representations of this type.%
\footnote{Standard treatments of probability theory assume countable additivity, but the motivation for this further requirement is reduced  given that we 
will not be assuming 
continuity axioms like Monotone Continuity.}

The applications of our aggregation theorems we provide in sections~\ref{S:harsanyi} and \ref{s:pooling} involve mixture-preserving and 
vector-measure representations with values in arbitrary partially ordered vector spaces. The next examples illustrate how such representations generalize some standard ways of representing incomplete or discontinuous preorders.

\begin{example}[Multi-representation]\label{ex:multi}
Suppose a preorder $\Geq$ on a set $X$ 
is represented in the following way. There is a family 
$\Set{f_i}{i\in\II}$ of functions 
from $X$ to $\RRR$ 
such that $x \Geq x'$ if and only if $f_i(x) \geq f_i(x')$ for all $i \in \II$. This can be rewritten as a representation by a single function $F \colon X \to \prod_{i\in\II}\RRR$ when $\prod_{i\in\II}\RRR$ is equipped with the product partial order,%
\footnote{In general, for a family $\{(X_i,\Geq_i):i\in\II\}$ of preordered sets, the product preorder $\GeqP$  on $\prod_{i\in\II} X_i$ is defined by the condition that $(x_i)_{i\in\II} \GeqP (y_i)_{i\in\II}$ if and only if $x_i \Geq_i y_i$ for each $i \in \II$.}
making it a partially ordered vector space, and $F$ is defined by $(F(x))_i=f_i(x)$. 
Note that if the $f_i$ are mixture preserving on convex $X$, then $F$ is also mixture preserving; and if the $f_i$ are 
probability measures 
on an algebra  $X$,  then $F$ is a vector measure.
\end{example}
Multi-representations like this 
are used in expected utility theory to represent the preferences of agents with incomplete values;%
\footnote{See e.g. \cite{rA1962, pF1982, SSK1995, SB1998, DMO2004, BS2008,  oE2008, MM2008, oE2014, GK2012, GK2013, MMT2017a}.}
in decision theory to represent the preferences of agents with incomplete beliefs;%
\footnote{See e.g. \citet{GS1989, SSK1995, tB2002, GMMS2003, rN2006, GMMS2010, OOR2012, GK2013}.}
and in probability theory to represent 
agents with incomplete beliefs.%
\footnote{For entries to a vast literature, see \citet{jH2003} and \citet{mS2016}; on multi-representation of incomplete comparative likelihood preorders, discussed further in section~\ref{s:pooling}, see \citet{dI1992} and \citet{AL2014}.}

\begin{example}[Lexicographic representation]\label{ex:lex}
Suppose a preorder $\Geq$ on a set $X$ 
is represented in the following way. There is a finite vector 
$(f_1, \dots, f_n)$ 
of functions from $X$ to $\RRR$ such that $x \Geq x'$ if and only if $f_j(x) = f_j(x')$ for all $j$ or $f_j(x) > f_j(x')$ for the least $j$ such that $f_j(x) \neq f_j(x')$. 
The $f_j$ can again be combined into a single function $F\colon X\to\RRR^n$, with $(F(x))_j=f_j(x)$.  As in Example \ref{ex:multi}, if the $f_j$ are mixture-preserving functions or probability measures, then $F$ is a mixture-preserving function or a vector measure. Moreover, this $F$ represents $\Geq$ if we endow $\RRR^n$ with the `lexicographic' order. To give the  general picture, let
$(\JJ, \Geq_{\JJ})$ be an ordered set, and let $\RRR^\JJ_\text{wo}$ be the subspace of $\RRR^\JJ$ whose members have well-ordered support.%
\footnote{That is, $\RRR^\JJ_\text{wo} = \Set{f \in \RRR^\JJ}{\Set{j \in \JJ}{f(j) \neq 0} \text{ is well-ordered by $\Geq_\JJ$}}$.} 
The lexicographic order $\Geq_{\text{lex}}$ on $\RRR^\JJ_\text{wo}$ is defined by the property that 
$f \Geq_{\text{lex}} f'$ if and only if $f = f'$ or $f(j) > f'(j)$ 
for the $\Geq_\JJ$-least $j$ such that $f(j) \neq f'(j)$. This makes $\RRR^\JJ_\text{wo}$ what \citet{HW1952} call a `lexicographic function space'. It is an ordered vector space. 
In the original example we can identity $\RRR^n$ with $\RRR^\JJ_\text{wo}$, for $\JJ\coloneqq\{1,2,\ldots,n\}$ with the usual ordering. 
\end{example}
Representations with values in lexicographic function spaces
are used in expected utility theory to represent the preferences of agents with discontinuous values;%
\footnote{See \citet{mH1954, pF1971, BBD1989, dB2016, HOR2016, MMT2017b}.}
in decision 
and game 
theory to represent the preferences of agents with discontinuous beliefs;%
\footnote{See \citet{BBD1991a, BBD1991b, BFK2008}.}  
and in probability theory to represent  
agents with discontinuous beliefs.%
\footnote{See \citet{jH2010, BH2018}.}

The constructions in these examples can be combined, yielding representations of possibly incomplete and discontinuous relations that take values in 
$\prod_{i \in \II}\RRR_{\text{wo}}^{\JJ_i}$, 
a partially ordered vector space where each component has the lexicographic order and their product has the product partial order.%
\footnote{\label{fn:matrix}If desired, the $\JJ_i$ can be enlarged so that the representation may be taken into $\prod_{i \in \II}\RRR_{\text{wo}}^{\JJ}$. Elements in this space can be seen as $\II \times \JJ$ matrices, with 
the row space lexicographically ordered, and one matrix ranking higher than another if it ranks higher in each row.} 
In fact, this construction is fully general: every partially ordered vector space can be  order-embedded in such a space.%
\footnote{That every ordered vector space can be order-embedded in a lexicographic function space was shown by \citet{HW1952}; 
the extension to partially ordered vector spaces is given in  \citet{HOR2016} and \citet{MMT2017b}.}

Finally, we note that our approach is also compatible with the use of representations with values in non-Archimedean ordered fields, such as non-standard models of the real-numbers.%
\footnote{\label{fn:fields}See e.g. \citet{pH1994a, pH1994b, pH1999, fH2009, jH2010, mP2014, BHW2018a, BHW2018b, BH2018}.}
These provide natural ways of representing discontinuous relations; see Example~\ref{ex:DR infinity} below. 
We discuss this topic further in section~\ref{S:literature}.

\begin{rem}[Preordered Vector Spaces]\label{rem:preorder} 
We focus on representations in partially ordered vector spaces, 
rather than merely preordered ones, 
because it seems desirable for objects that are ranked equally to be assigned \emph{the same} value. 
But for some purposes it is useful to consider 
preordered vector spaces instead. Our results have implications for  this more general case, 
insofar as  any preordered vector space has a partially ordered vector space as a quotient.%
\footnote{That is, if $\Geq_V$ is a linear preorder on $V$, then it determines a linear partial order on $V/{\Eq_V}$; the quotient map $V\to V/{\Eq_V}$ is a representation of $\Geq_V$.} 
Lemma~\ref{L:expectational} and 
Theorem~\ref{T:expectH} illustrate some of the extensions this makes available. 
\end{rem}

\subsection{Framework and axioms}\label{s:Framework}
We assume throughout that $X$ and $\III$ are nonempty, possibly infinite, sets. 
For each $i \in \III \, \cup \, \{0\}$, 
let $\Geq_i$ be a preorder on $X$. 
In applications, $\III$ is typically a population, 
and for each $i$ in $\III$, 
$\Geq_i$ expresses the values or beliefs of individual $i$, 
while $\Geq_0$ expresses those of the social observer.

Say that a family $\{f_i \colon X \to V_i\}_{i\in\II}$ of functions with values in vector spaces is {\em co-convex} if their joint range $(f_i)_{i\in \II}(X)\subset \prod_{i\in \II} V_i$ is convex. 
Equivalently, the $f_i$ are co-convex if there is a surjective map $q$ from $X$ onto a convex set 
$\bar{X}$, and mixture-preserving functions $\bar{f_i}\colon \bar{X}\to V_i$, such that $f_i=\bar{f_i} \circ q$. 
Thus co-convexity is a modest generalisation of the assumption that $X$ is a convex set and the $f_i$ are mixture preserving.

The central assumption of our main theorems will be that for $i \in \III \cup \{0\}$, the $\Geq_i$ have co-convex representations 
$f_i$ with values in partially ordered vector spaces $V_i$. 
 In the special case in which $\III$ is finite and each $V_i$ is the real numbers with the usual ordering, 
the requirement of co-convex representations was introduced by \citet{DMM1995}. 

Any preordered set can be shown to have a representation with values in  some  partially ordered vector space \citep[Lemma A.6]{MMT2018}, so the question is the significance of co-convexity.  
We postpone this question, though, until sections~\ref{S:harsanyi} and \ref{s:pooling}, where we present contexts in which co-convex representations naturally arise. 

We consider the following Pareto-style axioms.%
\footnote{See e.g. \citet{jW1993, jW1995} for discussion in the context of Harsanyi's theorem.}
\begin{quote}
\begin{enumerate}[label=\textup{P\arabic*}, nosep]
\item\label{P1} If $x \Eq_i y$ for all $i \in \III$, then $x \Eq_0 y$.
\item\label{P2} If $x \Geq_i y$ for all $i \in \III$, then $x \Geq_0 y$.
\item\label{P3} If $x \Geq_i y$ for all $i \in \III$, 
and $x \sGeq_j y$ for some $j$, then $x \sGeq_0 y$.
\item\label{P4} If $x \Geq_i y$ for all $i \in \III \setminus \{j\}$ and $x \Incom_j y$, then 
$x \nLeq_0 y$.
\end{enumerate} 
\end{quote}

When $\III$ is finite
and the $\Geq_i$ are preference relations, 
\ref{P1} is Pareto indifference; 
\ref{P1}--\ref{P2} together are sometimes called semi-strong Pareto; and \ref{P1}--\ref{P3} are strong Pareto. Since we are allowing for incompleteness, \ref{P4} is a natural supplement. 
Given \ref{P2} and a sufficiently rich domain, \ref{P4} follows from the simpler  condition that,  if $x\Eq_i y$ for all $i\in\III\setminus\{j\}$ and $x\Incom_j y$, then $x\Incom_0 y$. But to avoid domain conditions, we will appeal to \ref{P4} as formulated. 
\ref{P1}--\ref{P4} are also natural conditions under other interpretations of the $\Geq_i$, such as when  they are comparative likelihood relations.

\subsection{Results}\label{s:Results} 
Suppose we are given co-convex representations $f_i\colon X\to V_i$ of the $\Geq_i$, $i\in \III$.  
Set $V_\III \coloneqq \ProdI V_i$ and $\fI\coloneqq (f_i)_{i\in\III}\colon X\to V_\III$. We are interested in the question of whether $\Geq_0$ has a representation of the form $L\fI$, with $L$ a linear map from $V_\III$ into some partially ordered vector space $V$.    
When $\III$ is finite, we can write 
\[
\textstyle{L\fI = \sum_{i\in\III} L_if_i}
\]
where $L_i\colon V_i\to V$ is the $i$th component of $L$;%
\footnote{In other words, regardless of whether $\III$ is finite, $L_i \coloneqq L1_i$, where $1_i$ is the natural embedding of $V_i$ into $V_\III$.}
thus $L\fI$ is an additive representation.
But when $\III$ is infinite, $L$ is not determined by its components; relatedly, the sum over $i$ does not make sense.%
\footnote{Nonetheless, the linearity of $L$ does entail a kind of finite additivity; see section \ref{S:harsanyi} for further discussion in the context of Harsanyi's theorem.}
In any case, by endowing $V_\III$ with the product partial order $\GeqP$, we can consider the positivity of $L$.  
Note that if $L$ is positive, or strictly positive, then so is every $L_i$, but  if $\III$ is infinite, the converse need not hold.

Since $\fI$ and $L\fI$ are automatically co-convex, a necessary condition for the existence of a representation 
of $\Geq_0$ 
of the form $L\fI$ is that there exists \emph{some} representation 
$f_0\colon X\to V_0$ 
of $\Geq_0$ 
such that all the $f_i$, including $f_0$, are co-convex. Our first theorem says that this necessary condition is also sufficient, and lays out the connection between the Pareto-style axioms and the positivity properties of $L$.

\begin{thm}\label{T:infinitaryH}
Assume that for $i \in \III \cup\{0\}$ 
the $\Geq_i$ have co-convex representations $f_i$.   
\begin{enumerate}[label=(\alph*), nosep]
\item\label{part:infinitaryH P1} \ref{P1} holds if and only if  $\Geq_0$ has a representation of the form $L\fI$, with $L$ linear. 
\item\label{part:infinitaryH P2}  \ref{P1}\!--\ref{P2} hold if and only if 
$\Geq_0$ has a representation of the form $L\fI$, 
with $L$ linear and positive. 
\item\label{part:infinitaryH P3} \ref{P1}\!--\ref{P3} hold if and only if 
 $\Geq_0$ has a representation of the form $L\fI$, with $L$ linear and strictly positive.  
\item\label{part:infinitaryH P4} \ref{P1}\!--\ref{P4} hold if and only $\Geq_0$ has a representation of the form $L\fI$, with $L$ linear and
strictly positive, 
  and every $L_i$ 
an order embedding.
\end{enumerate}
\end{thm}

An obvious question is whether \emph{every} representation $f_0\colon X\to V_0$ of $\Geq_0$, co-convex with $\fI$, has the form $L\fI$.   Note that, if $f_0$ is a representation, then so is
$f_0+b$, for any constant $b\in V_0$, but at most one of these will be of the form $L\fI$. So a more reasonable question is whether we can write every $f_0$ in the form $L\fI+b$. 
The following result gives an affirmative answer, subject to the following domain richness condition.  
  
\begin{quote}
\begin{enumerate}[label=\textup{DR}, nosep]
\item \label{DR} $V_\III = \Span(\fI(X)-\fI(X))$. 
\end{enumerate}
\end{quote}
In other words, $\fI(X)$ is not contained in any proper affine subspace of $V_\III$; 
equivalently, $V_\III$ is the affine hull of $\fI(X)$.%
\footnote{\label{fn:DR}To interpret \ref{DR}, it is worth noting that, by Lemma \ref{L:Vspan} below, and taking into account the fact that $\fI(X)$ is convex, $\Span(\fI(X)-\fI(X))=\{\lambda(\fI(x)-\fI(y)): \lambda>0, x,y\in X\}$. 
Thus when the $\Geq_i$ are preference relations, so that the $V_i$ may be seen as utility spaces, 
\ref{DR} says that any logically possible profile of 
utility differences (i.e. any element of $V_\III$) is realized as the utility difference between some $x,y\in X$, at least up to scale. Indeed,  \ref{DR} is equivalent to the conjunction of the claims 
(i) that 
$V_i=\Span(f_i(X)-f_i(X))$, so that 
that none of the utility spaces $V_i$ is gratuitously large 
(in terminology introduced in section~\ref{S:uniqueness}, this means that $f_i$ is pervasive); 
and (ii) that $\Span(\fI(X)-\fI(X))=\ProdI \Span(f_i(X)-f_i(X))$.}

\begin{thm}\label{T:infinitaryH-2} 
Assume that for $i \in \III \cup\{0\}$
the $\Geq_i$ have co-convex representations $f_i$. 
For the left-to-right directions of parts~\ref{part:infinitaryH-2 P2}--\ref{part:infinitaryH-2 P4}, assume \ref{DR}. 
\begin{enumerate}[label=(\alph*), nosep]
\item\label{part:infinitaryH-2 P1} \ref{P1} holds if and only if $f_0=L\fI+b$ for some $b\in V_0$ and $L$ linear. 
\item\label{part:infinitaryH-2 P2} \ref{P1}\!--\ref{P2} hold if and only if 
$f_0=L\fI+b$ for some $b\in V_0$ and $L$ 
linear and positive. 
\item\label{part:infinitaryH-2 P3} \ref{P1}\!--\ref{P3} hold  if and only if 
$f_0=L\fI+b$ for some $b\in V_0$ and $L$  
linear and strictly positive. 
\item\label{part:infinitaryH-2 P4} \ref{P1}\!--\ref{P4} hold if and only if  $f_0=L\fI+b$ for some $b\in V_0$ and $L$ 
linear and strictly positive,  with 
every $L_i$ an order embedding.
\end{enumerate}
\end{thm} 
 \noindent In the special case in which $\III$ is finite and $V_i =\RRR$ for each $i\in \III \cup \{0\}$ (so that \ref{P4} is vacuous), this is proved in \citet[Prop. 1]{DMM1995} without assuming \ref{DR}.%
\footnote{In this special case, our methods give an alternative and arguably simpler \ref{DR}-free proof of part~\ref{part:infinitaryH-2 P1}, but they do not appear to help at all with \ref{DR}-free proofs of the other parts. 
See however the discussion at the end of section \ref{S:domain}.}

So far we have allowed the representations $f_i$ to have values in different partially ordered vector spaces.  One might wish to construct a single partially ordered vector space $V$ in which all of the different preorders are represented, 
and in which the representation of $\Geq_0$ is essentially the sum of the representations of the $\Geq_i$.  `Essentially' is required here since we cannot literally sum over $\III$ when it is infinite. But say that a map $S\colon \ProdI V\to V$ \emph{extends summation} if it restricts to the summation map $\SumI V\to V$; 
equivalently, each component $S_i$ of $S$ is the identity map on $V$.

\begin{thm}\label{T:infinitaryH-3} 
Assume that for $i \in \III \cup\{0\}$ 
the $\Geq_i$ have co-convex representations $f_i$.   
Then \ref{P1}\!--\ref{P4} hold if and only if there exists a partially ordered vector space $V$,  representations $g_i\colon X\to V$ 
of the $\Geq_i$ for $i \in \III \cup\{0\}$, 
and a linear map $S\colon\ProdI V\to V$ such that 
\begin{enumerate}[label=(\alph*), nosep]
    \item \label{part:H-3a} $g_0=S\gI$

    \item \label{part:H-3b} $S$ extends 
    summation and is strictly positive on 
    $\ProdI \Span (g_i(X)-g_i(X))$%
\footnote{Here positivity refers to the product partial order on $\ProdI V$, restricted to 
 $\ProdI \Span (g_i(X)-g_i(X))$.}    
\item \label{part:H-3c} $\fI$, $f_0$, $\gI$, and $g_0$ are 
together 
co-convex.
\end{enumerate}
If \ref{DR} holds as well as \ref{P1}--\ref{P4}, we can further require $V=V_0$ and $g_0=f_0$. \end{thm}

\subsubsection{Domain assumptions}\label{S:domain}

\ref{DR} is a stronger domain assumption than is needed for 
Theorem~\ref{T:infinitaryH-2}\ref{part:infinitaryH-2 P2}--\ref{part:infinitaryH-2 P4} 
and the last claim in Theorem \ref{T:infinitaryH-3},%
\footnote{\label{fn:weaken DR}For example, for Theorem \ref{T:infinitaryH-2}\ref{part:infinitaryH-2 P2}\ref{part:infinitaryH-2 P3}, \ref{DR} 
could be replaced by the weaker assumption that $\Span(\fI(\X)-\fI(\X))$ contains the positive cone 
$\Set{v\in V_\III}{v \GeqP 0}$; 
for part \ref{part:infinitaryH-2 P4} 
and for 
the last part of 
Theorem~\ref{T:infinitaryH-3}
one could assume that 
$\Span(\fI(\X)-\fI(\X))$ contains both the positive cone and 
 $\SumI V_i$. 
 However, even these assumptions are stronger than necessary.}
 and we adopt it for its simplicity. It cannot simply be dropped from 
these theorems,
even when $\III$ is finite, 
 as the following example shows.

\begin{example}\label{ex:DR counterexample}
Let $X = \RRR$ and $\III = \{1,2\}$. Let $\Geq_1$ and $\Geq_2$ be represented by $f_1(x) = x$ and $f_2(x) = -x$. 
Thus $V_\III = \RRR \times \RRR$, $\fI(x)=(x,-x)$, and $\Span(\fI(X)-\fI(X)) = \RRR(1,-1)$. 
Let $V_0 = \RRR$ with the `trivial' linear partial order such that any two distinct elements are incomparable (thus $x \Geq_{V_0} 0$ if and only if $x=0$).  Let $\Geq_0$ be represented by $f_0(x) = x$.  
Thus $(\fI,f_0)(X)$ is convex, $\Geq_0$ satisfies $\ref{P1}$--$\ref{P4}$, but \ref{DR} fails. 

Suppose $f_0 = L\fI +b$, with $L \colon V_\III \to V_0$ positive.  
We have  $(1,0) \GeqP 0$ and since $L$ is positive,  $L((1,0))= 0$; and similarly, $L((0,1)) = 0$. Linearity of  $L$ implies $L= 0$, contradicting $f_0 = L\fI +b$. 
Thus parts \ref{part:infinitaryH-2 P2}--\ref{part:infinitaryH-2 P4} of Theorem~\ref{T:infinitaryH-2} do not hold. 
Similarly for 
the last claim of
Theorem \ref{T:infinitaryH-3}: there is no representation 
$g_1$ 
of $\Geq_1$ 
with values in $V_0$.
\end{example}

When $\III$ is infinite, \ref{DR} is perhaps surprisingly strong.

\begin{example}\label{ex:DR infinity}
Suppose that $\III =  \NNN$,  $X = \ProdI[0,1]$, $V_i=\RRR$, $f_i(x) = x_i$ for $i\in\III$. Then $V_\III = \ProdI \RRR$ is the space of sequences of real numbers, 
but $\Span(\fI(X)-\fI(X))$ is the subspace of bounded sequences,  
so \ref{DR} fails. 

Theorem~\ref{T:infinitaryH-2}\ref{part:infinitaryH-2 P2}--\ref{part:infinitaryH-2 P4}  also fail. 
As in \citet{BHW2018a}, choose a uniform probability measure $\mu$ on $\NNN$ with values in a non-Archimedean ordered field $F$ 
that extends the reals; as such, it is an ordered real vector space. Let $V_0$ be the ordered real vector space of finite (including infinitesimal) elements of $F$; thus for every  $x$ in $V_0$ there is some natural number $n$ with $n\Geq_{V_0} x$. 
Note that $\mu$ has values in $V_0$. Let $f_0 \colon X \to V_0$ map each element of $X$ to its $\mu$-expectation; that is $f_0(x) 
= \sum_{i \in \III} \mu(\{i\})x_i$. The $\Geq_0$ so represented satisfies \ref{P1}--\ref{P4} (note that the $\Geq_i$ are complete, so \ref{P4} is vacuous).  
All the $\Geq_i$ satisfy strong independence, but $\Geq_0$ violates both the Archimedean and mixture continuity axioms of expected utility.

Suppose $f_0 = L\fI +b$ for some linear mapping $V_\III \to V_0$, implying $b = 0$. 
Consider the sequence 
$v = (1,2,3,4, \dots)$ in $V_\III$, and for each natural number $n$, the bounded sequence $v_n = (1,2, \dots, n-1, n,n,n,n, \dots)$. 
By linearity of $L$ we must have $L(v_n) \Geq_{V_0} n-1$  (since $f_0(\tfrac 1n v_n)$ is infinitesimally close to $1$). 
But  $v \GeqP v_n$ for every natural number $n$. So if $L$ is positive, we must have $L(v) \Geq_{V_0} n$ for every $n$. 
But that is impossible, by construction of $V_0$. 
A similar argument shows that the last claim of Theorem~\ref{T:infinitaryH-3}, involving \ref{DR}, also fails.
\end{example}

An interesting question therefore is how far \ref{DR} can be weakened in Theorems~\ref{T:infinitaryH-2} and \ref{T:infinitaryH-3}. 
We plan to take up that question in other work, but for now we note that one can sometimes bypass \ref{DR} by allowing $L$ in Theorem \ref{T:infinitaryH-2} to be defined only on a subspace of $V_\III$ that contains $\fI(X)$. 
For example, by adding constants to the $f_i$, we can assume that $\fI(X)$ is contained in $Y \coloneqq \Span(\fI(X)-\fI(X))$. 
Without any need for \ref{DR},    Theorem~\ref{T:infinitaryH-2}\ref{part:infinitaryH P1}--\ref{part:infinitaryH P3} hold if $L$ is only required to be defined on $Y$. 
Part  \ref{part:infinitaryH P4} also holds under the further assumption that $Y$ contains $\SumI V_i$, which is needed for the components $L_i$ to be defined.%
\footnote{The proofs of these claims are trivial variations on the proof of Theorem~\ref{T:infinitaryH-2}, so we omit them. 
The situation for the last statement in Theorem \ref{T:infinitaryH-3} is slightly less straightforward, but, in short,  we can again replace \ref{DR} by the assumption that $Y$ contains $\SumI V_i$, if we only require in part \ref{part:H-3b} that $S$ is 
strictly 
positive on $Y'\coloneqq \Span(\gI(X)-\gI(X))$.
That is to say, assuming that $Y$ contains $\SumI V_i$, 
\ref{P1}--\ref{P4} hold if and only if there exist representations $g_i\colon X\to V_0$ of $\Geq_i$ for $i\in\III$, and a linear map $S\colon \ProdI V_0\to V_0$, such that $f_0=S\gI$, $S$ extends summation, $S$ is 
strictly   
positive on $Y'$, and $\fI$, $f_0$, and $\gI$ are together co-convex. 
} 
This last assumption is weaker than \ref{DR} when $\III$ is infinite, and 
is satisfied in Example~\ref{ex:DR infinity}. 
To illustrate, $f_0$ in that example extends uniquely to a strictly positive linear $L \colon Y \to V_0$ 
that maps each bounded sequence to its $\mu$-expectation, 
with every $L_i$ an order embedding, and we have $f_0 = L\fI$.  
But the example shows that $L$ cannot be extended to a positive linear map $V_\III \to V_0$.

We will not pursue domain questions any further. 
For simplicity, we use \ref{DR} to discuss the uniqueness of the representations discussed section \ref{s:Results}, to which we now turn. But we avoid it in the applications we present in section~\ref{S:applications}.

\subsection{Uniqueness}\label{S:uniqueness}

We first address the general question of to what extent representations with values in partially ordered vector spaces are unique. Say that a function $f \colon X \to V$ is {\em pervasive} if $V = \Span(f(X)-f(X))$; equivalently, $V$  is the affine hull of $f(X)$. The restriction to pervasive representations in the next result is mild, since by adding a constant to $f$ we can always obtain a pervasive representation of the preorder on $X$ that is represented by $f$: for any $x_0\in X$, the representation $f^*\colon X\to \Span(f(X)-f(X))$ defined by $f^*(x)=f(x)-f(x_0)$ is pervasive.

\begin{lemma}\label{L:coconvex0}
Suppose given co-convex representations $f\colon X\to V$ and $g\colon X\to V'$ of a preorder $\Geq$ on $X$. Suppose that $f$ and $g$ are pervasive. 
Then there exists a unique linear order isomorphism $L\colon V\to V'$ and unique $b \in V'$ such that $g=Lf+b$. 
\end{lemma}

The following result explains the sense in which the type of representation of $\Geq_0$ discussed in Theorem \ref{T:infinitaryH} is unique.

\begin{prop}\label{prop:uniqueness1}
Assume that for $i\in\III\cup\{0\}$ the $\Geq_i$ have co-convex representations $f_i$ such that DR holds.
Suppose that $L\colon V_\III\to V$  and $L'\colon V_\III\to V'$ are linear maps to partially ordered vector spaces such that $L\fI$ and  $L'\fI$  represent $\Geq_0$. 
Then there is a unique linear order isomorphism 
$M\colon L(V_\III)\to L'(V_\III)$ such that $L'=ML$. 
 \end{prop}
\noindent
The proof establishes in part that  
$L\fI\colon X\to L(V_\III)$ and $L'\fI\colon X\to L'(V_\III)$ 
are co-convex, pervasive representations of $\Geq_0$, so that Lemma~\ref{L:coconvex0} applies. 
Note that Proposition \ref{prop:uniqueness1}, unlike Theorem \ref{T:infinitaryH}, assumes \ref{DR}. Here is an example to illustrate why.
\begin{example}
Let $X=\RRR$, let $\III=\{1,2\}$, and let $\Geq_0,\Geq_1,\Geq_2$ all equal the standard ordering on $\RRR$. Let $V_1,V_2=\RRR$ with the standard ordering, and let $f_1,f_2\colon X\to\RRR$ both be the identity map. 
Thus 
$V_\III = \RRR \times \RRR$ 
and DR fails. Let $L\colon V_\III \to V\coloneqq \RRR$ map $(x,y)\mapsto x+y$, 
and let 
$L'\colon V_\III \to V'\coloneqq \V_\III$ 
be the identity map. 
Then $L\fI$ and $L'\fI$ both represent $\Geq_0$, but there is no linear map $M\colon V\to V'$ such that $L'=ML$.
\end{example}

Next we establish that the $L$ and $b$ in Theorem \ref{T:infinitaryH-2}  are unique.
\begin{prop}\label{prop:uniqueness2} 
Assume that for $i\in\III\cup\{0\}$ the $\Geq_i$ have co-convex representations $f_i$ such that DR holds.   Then there at most one 
linear map $L\colon V_\III\to V_0$ and one $b\in V_0$ such that $f_0=L\fI+b$.
\end{prop}
\noindent Here, \ref{DR} cannot be dropped, even in the special case in which $\III$ is finite and $V_i = \RRR$ for each $i\in \III \cup \{0\}$, by \citet[Cor. 1]{pF1984}.

In Theorem~\ref{T:infinitaryH-3},  
given \ref{DR}, we can choose $g_0$ to be a pervasive representation of $\Geq_0$. (Choose $f_0$ to be pervasive, using the construction before Lemma~\ref{L:coconvex0}, and then apply the last statement of Theorem \ref{T:infinitaryH-3}.) Thus we assume that $g_0$ and $g'_0$ are pervasive in the following result.
\begin{prop}\label{prop:uniqueness3} Assume that for $i\in\III\cup\{0\}$ the $\Geq_i$ have co-convex representations $f_i$ such that DR holds. Suppose $g_i\colon X\to V$ and $S\colon \ProdI V\to V$ satisfy the conditions \ref{part:H-3a}--\ref{part:H-3c} of Theorem \ref{T:infinitaryH-3}, as do
$g'_i\colon X\to V'$ and $S'\colon \ProdI V'\to V'$.
Assume also that $g_0$ and $g'_0$ are pervasive. 
Then  there exists a unique linear order isomorphism $L\colon V\to V'$ and unique constants $b_i\in V'$ such that,
      for all $i\in\III\cup\{0\}$,  $g'_i=Lg_i+b_i$.
      Moreover, $S'(v)=LS((L\inv v_i)_{i\in\III})$ for all $v\in \SumI V'+ 
      \Span(\gI'(X)-\gI'(X))$.
 \end{prop}

\section{Applications}\label{S:applications}
We now present two applications in which co-convex  representations naturally arise.  

\subsection{Preference aggregation}\label{S:harsanyi}

Let us assume that the $\Geq_i$ are preference relations, with $\Geq_0$ that of the social observer. In this context, co-convex representations arise naturally from the main expected utility axiom of strong independence 
(see note~\ref{fn:SI}).

\begin{lemma}\label{L:SI} Suppose $X$ is a convex set. 
\begin{enumerate}[label=(\roman*), nosep]
\item \label{part:SIa} A preorder on $X$ satisfies strong independence if and only if it has a mixture-preserving representation. The representation may be chosen to be pervasive.
\item \label{part:SIb} Any family 
$\{f_i \colon X \to V_i\}_{i\in\II}$
of mixture-preserving functions is co-convex. 
\end{enumerate}
\end{lemma} 
Given that the $\Geq_i$ satisfy strong independence, 
pervasive (and, by Lemma~\ref{L:coconvex0}, essentially unique) mixture-preserving representation 
$f_i$ may be chosen for each of them, and these representations are automatically co-convex. 
The results of section \ref{S:main} therefore apply.%
\footnote{%
For applications of Theorem \ref{T:infinitaryH-2}, it may be useful to note that, in this context, \ref{DR} is equivalent to the following domain richness condition, stated directly in terms of $X$ and the $\Geq_i$: 
\begin{enumerate}[label=\textup{DR$'$}, nosep]
\item \label{DR'} Suppose given $x_i,y_i\in X$ and $\ll_i>0$, for each $i\in\III$. Then there exist $z,w\in X$ and $\ll>0$ such that, for all $i\in \III$, 
 $\frac{\ll}{\ll+\ll_i} w +
 \frac{\ll_i}{\ll+\ll_i}
 x_i\Eq_i 
 \frac{\ll}{\ll+\ll_i}z+
 \frac{\ll_i}{\ll+\ll_i}
y_i$. 
\end{enumerate}%
(Heuristically: the difference in value 
for $i$ between $z$ and $w$ is $\ll_i/\ll$ times the difference between $x_i$ and $y_i$.)}
Here, for example, is a corollary of Theorem \ref{T:infinitaryH-3}. 
\begin{thm}\label{T:preferenceH}
Suppose $X$ is convex, 
and for all 
$i \in \III \cup \{0\}$, 
$\Geq_i$ satisfies strong independence. 
Assume \ref{P1}--\ref{P4}. Then there exists a partially ordered vector space $V$, mixture-preserving representations 
$f_i\colon X\to V$ 
of the $\Geq_i$, and a linear map $S\colon \ProdI V\to V$ that extends summation, such that 
$f_0= S\fI$.
\end{thm}

In this context, our results generalize the celebrated social aggregation theorem of \citet{jH1955}.  In his framework, 
the population $\III$ is finite, 
and $X$ is the set of simple probability measures on a given set of social outcomes. 
In contrast, we allow $\III$ to be infinite, and we allow $X$ to be an arbitrary convex set, which may in particular be infinite dimensional.%
\footnote{The dimension of $X$ is, by definition, the dimension of the smallest affine space that contains it; equivalently, the dimension of the vector space $\Span(X-X)$.} 
This allows for a wide range of models of uncertainty. 
For example, in the setting of objective risk, 
$X$ may be any convex set of probability measures on a measurable space;  in the setting of objective risk and subjective uncertainty, it may be the set of Anscombe-Aumann acts;  
in the setting of subjective uncertainty, it may be the set of Savage acts when those are equipped with convex structure,  as in for example \citet{GMMS2003};%
\footnote{As is well known, however, allowing for subjective uncertainty is likely to lead to impossibility results; see the end of section~\ref{S:literature}.}
it may be a set of simple lotteries with nonstandard probabilities; 
or it may be an arbitrary 
mixture space, 
and hence isomorphic to a convex subset of a vector space as noted by \citet{mH1954}.%
\footnote{For a discussion of the relationship between mixture spaces in the sense of Hausner and mixture \emph{sets} in the sense of \citet{HM1953}, see \citet{pM2001}; Mongin constructs a natural map from each mixture set onto a convex set, and shows that it is an isomorphism if and only if the mixture set is a mixture space (or in his terminology `non-degenerate').}

One version of Harsanyi's result is that, if \ref{P1}--\ref{P3} hold, as well as strong independence, continuity and completeness for each $\Geq_i$, then the social preorder can be represented by the sum of real-valued mixture-preserving individual utility functions.%
\footnote{For discussion of different variations of Harsanyi's result, see e.g. \citet{jW1993}.
Note that Harsanyi's result is usually stated in terms of a \emph{weighted} sum of individual utility functions; but one can absorb the weights into the utility functions to get an unweighted sum.
Of course, in both the weighted and unweighted versions, the representation of the social preorder is a linear transformation of the profile of individual utilities. 
Roughly speaking, Theorems \ref{T:infinitaryH} and \ref{T:infinitaryH-2} generalize the formulation with weighted sums, while Theorem~\ref{T:infinitaryH-3} and its corollaries Theorems~\ref{T:preferenceH} 
and \ref{T:expectH} 
generalize the unweighted version.
}
Theorem~\ref{T:preferenceH} shows that essentially the same conclusion holds without assuming continuity or completeness, but assuming \ref{P4} (which is vacuous given completeness),  and allowing the utility functions to be vector-valued. 
In light of the discussion after Example \ref{ex:lex}, 
and especially note~\ref{fn:matrix}, 
the conclusion of  Theorem~\ref{T:preferenceH} may be made visibly closer to Harsanyi's by taking utility values to be matrices of real numbers; Harsanyi's conclusion is then the case where the matrices are one-by-one.

The only caveat is that, when $\III$ is infinite, the linear map $S$ used to combine the individual utility functions 
$f_i$ 
does not simply sum them up. However, we still get separability, or `finite additivity'. 
If $\III=\JJJ\sqcup\KKK$, then $V_\III=V_\JJJ \times V_\KKK$; denoting restrictions in the obvious way, we have 
$S\fI = S_\JJJ f_\JJJ+ S_\KKK f_\KKK$. 
So much follows from the linearity of $S$; since $S$ also extends summation, $S_\JJJ$ is simply the summation map when $\JJJ$ is finite.

Our results also yield fully additive representations in some important cases when $\III$ is infinite.   
Suppose, for example, that in the context of Theorem~\ref{T:preferenceH}, each element of $X$ is a gamble in which only finitely many people from $\III$ have any chance to exist. For each $i\in\III$ and $x,x'\in X$, it is natural to suppose that $x\Eq_i x'$ if $i$ is certain not to exist in either one; thus 
$f_i(x)=f_i(x')$.  By subtracting a constant from   $f_i$,  we may assume that  $f_i(x)=f_i(x')=0$.   (Note that renormalizing $f_i$ in this way changes neither the fact that $f_i$ represents $\Geq_i$ nor the fact that $S\fI$ represents $\Geq_0$.)  The upshot of this construction is that,  for each $x\in X$, $f_i(x)=0$ for all but finitely many $i\in\III$. We can therefore write $S\fI=\sum_{i\in\III} f_i$, a fully additive representation of the social preorder.

Suppose now that $X$ is a convex set 
of probability measures on a measurable space $Y$.  
It is natural to ask whether 
the representations constructed in Theorem~\ref{T:preferenceH},  for example,  can be written as integrals over $Y$, in the style of expected utility theory. 

To make the question precise, suppose we are given a vector space $V$ and a separating vector space $V'$ of linear functionals $V\to \RRR$.%
\footnote{Endowing $V$ with the weak topology 
with respect to $V'$ makes it a locally convex topological vector space whose dual is $V'$ \citep[3.10]{wR1991}.}
A function  $U\colon Y\to V$ is 
\emph{weakly $X$-integrable with respect to $V'$} 
if there exists $f\colon X \to V$ such that $\int_Y \Lambda \circ U \, \mathrm{d} \mu = \Lambda\circ f(\mu)$ for all $\Lambda \in V', \mu \in X$.
In particular, every $\Lambda\circ U$ must be 
Lesbesgue integrable against every $\mu\in X$. 
The {\em Pettis} or {\em weak} integral is defined by setting $\int_Y U \, \mathrm{d} \mu \coloneqq f(\mu)$. 
When $f \colon X\to V$ can be written in this form, 
we say that $f$ is {\em expectational}. 
The question, then, is whether the representations $f_i$, including $f_0$, can be chosen to be expectational functions.

In the most common case, where  $X$ is a convex set of \emph{finitely supported} probability measures on a measurable space $Y$ with measurable singletons,  there is a straightforward positive answer: \emph{any} mixture-preserving, vector-valued function on $X$ is expectational (independently of how $V'$ is chosen).  However, in the general case, it turns out that we need to consider representations with values in {\em preordered} vector spaces (cf.~Remark \ref{rem:preorder}). 
Indeed, the following result \citep[Lemma 4.3]{MMT2018} contrasts with Lemma \ref{L:SI}\ref{part:SIa}.
\begin{lemma}\label{L:expectational}
Let $X$ be an arbitrary convex set of probability measures. A preorder on $X$ satisfies strong independence if and only if it has an expectational (and not merely mixture-preserving) representation with values in a preordered (but not necessarily partially ordered) vector space.
\end{lemma}

The next result uses this to elaborate on Theorem~\ref{T:infinitaryH-3}.

\begin{thm}\label{T:expectH}
Suppose $X$ is a convex set of probability measures on a measurable space $Y$, 
and, for all 
$i \in \III \cup \{0\}$, 
$\Geq_i$ satisfies strong independence. 
Assume \ref{P1}--\ref{P4}. 
Then there exists a preordered vector space $V$ equipped with   
a separating vector space $V'$ of linear functionals;  
for $i \in \III \cup \{0\}$, 
functions $U_i \colon Y \to V$ 
that are  weakly $X$-integrable  with respect to $V'$, 
such that $\mu \mapsto \int_Y U_i \,\mathrm{d}\mu$ represents $\Geq_i$; 
and a linear map $S \colon\ProdI V\to V$ that extends summation 
such that $U_0 = S\UI$.
\end{thm}
\noindent Thus when $\III$ is finite, for example, without assuming continuity or completeness or interpersonal comparisons, 
we find that the social preorder is represented by the mapping $\mu \mapsto \int_Y \sum_{i\in\III} U_i  \, \mathrm{d}\mu$; 
in other words, by expected total utility.

\subsection{Opinion pooling}\label{s:pooling}

Here we assume that $X$ is a Boolean algebra of events, 
understood as sets of states of nature, 
and the $\Geq_i$ are comparative likelihood relations on $X$, expressing the beliefs of each individual, with $\Geq_0$ expressing those of the social observer. So for events $A, B \in X$, $A \Geq_i B$ has the interpretation that, according to $i$, $A$ is at least as likely as $B$. 
 
Generalising ordinary probability measures, 
we will be considering representations by vector measures,%
\footnote{These were defined following Example~\ref{ex:sphere}.}
with values in partially ordered vector spaces. 
In Lemma~\ref{L:SI2}\ref{part:SI2-a} below, we give a necessary and sufficient condition for a preorder $\Geq$ on $X$ to be representable by a vector measure $f$. 
But in terms of standard axioms of comparative probability \citep[see e.g.][]{AL2014}, 
the existence of such a representation entails that $\Geq$ satisfies
Reflexivity, Transitivity, and Generalized Finite Cancellation. 
Positivity and Non-Triviality are equivalent respectively to the further conditions that $f(A)\Geq_V 0$ for all $A\in X$ and that $f(A)\sGeq_V 0$ for some $A\in X$. We treat Positivity and Non-Triviality as natural but optional assumptions about $\Geq$, rather than imposing these conditions on $f$.   The main remaining axioms---Completeness and Monotone Continuity---are  the ones which our use of vector measures is 
intended to avoid.

Our Pareto-style axioms \ref{P1}--\ref{P4} 
may seem as plausible here as in the context of preference aggregation. 
However, consider
\begin{example}\label{ex:urn}
Let $\III = \{1,2\}$. A ball is going to be drawn randomly from an urn containing three balls, red, yellow and blue. Individual $1$ privately observes that the ball is not red, and concludes $\{B\} \Eq_1 \{R, Y\}$. Individual $2$ privately observes that the ball is not yellow, and concludes $\{B\} \Eq_2 \{R, Y\}$. The social observer, privy to each individual's private information, will conclude $\{B\} \sGeq_0 \{R, Y\}$, contrary to \ref{P1}. 
\end{example}
The natural reply, though, is that \ref{P1}--\ref{P4} are not designed for the problem in which the observer's task is to aggregate the opinions represented by $(\Geq_i, \II_i)$, where $\II_i$ is $i$'s private information. They are meant for circumstances in which the only data the observer has, or considers relevant, is the $\Geq_i$, say when all private information is either hidden from the observer, or has been made common knowledge. Thus we assume that \ref{P1}--\ref{P4} are applicable in at least an important range of cases; see~\citet{DL2016} for related discussion. 

Linear opinion pooling is the idea that society's beliefs should be 
represented by a linear combination of individual beliefs with nonnegative coefficients, 
and goes back at least to \citet{mS1961}. When each $\Geq_i$ can be represented by a probability measure on $X$, linear pooling was axiomatized by \citet{kM1981}.%
\footnote{In this special case, the linear combination is typically normalized to a convex combination.}
An alternative 
axiomatization using 
Pareto-style conditions 
was given in \citet{pM1995} and \citet{DMM1995}, applying the Lyapunov convexity theorem.  
Further references are given in section~\ref{S:literature}. 
Here we give linear pooling results that replace ordinary probability measures 
with vector measures. 
We will first extend the approach based on Lyapunov's theorem. However, this requires a 
finite population, as well as some technical restrictions, 
so we will go on to 
explain 
an alternative approach,
still using 
Pareto-style conditions, 
that mixes objective and subjective probability, in the spirit of Anscombe-Aumann decision theory.

\subsubsection{Lyapunov} \label{s:Lyapunov}
Our aggregation theorems apply if we assume that 
the $\Geq_i$ are represented by \emph{co-convex} vector measures.
This assumption follows under certain conditions from the Lyapunov convexity theorem.
The following version is proved in \citet{AP1981} (their Theorem 2.2, which allows $X$ more generally to be an $F$-algebra). 
\begin{thm}[Lyapunov]\label{T:Lyapunov}
Suppose $X$ is a $\sigma$-algebra, and 
$\{f^i\colon X\to \RRR\}_{i\in\II}$ is a finite family of finitely additive, bounded, non-atomic
signed measures on $X$. Then the $f^i$ are co-convex. 
\end{thm}

With this in mind, we say that a vector measure $f\colon X\to V$ is \emph{admissible} if $V$ is finite-dimensional, and, with respect to some
(hence any) basis of $V$, each component of $f$ is bounded and non-atomic.%
\footnote{Each component of $f$ is automatically a finitely additive signed measure on $X$. Following \citet{AP1981}, a finitely additive signed measure $f^j$ is said to be non-atomic if for every $\epsilon>0$ there is a finite partition $\{A_1,\ldots,A_n\}$ of $X$ such that for all $k$ the total variation $|f^j|(A_k)$ 
of $A_k$ under $f^j$ 
is less than $\epsilon$. 
Note that non-atomicity in this sense has nothing to do with the partial order on $V$.} 
To illustrate, these conditions are fulfilled by $f$ in Example \ref{ex:sphere}. 
Then Lyapunov's theorem ensures that $f(X)$ is convex; it also ensures that any finite number of admissible vector measures are co-convex.  
 
We thus obtain (for example) this application of Theorem~\ref{T:infinitaryH}.
\begin{thm}\label{T:pooling-c}
Suppose that $\III$ is finite, and that, for $i\in\III\cup\{0\}$, $\Geq_i$ is a preorder on a 
$\sigma$-algebra $X$ 
that can be represented by an admissible vector measure $f_i\colon X\to V_i$.  Suppose also that  \ref{P1}--\ref{P4} hold. 
Then there exists a 
finite dimensional 
partially ordered vector space $V$ and, for each $i\in\III$, a linear order embedding $L_i\colon V_i\to V$, 
such that 
the admissible 
vector measure
$\sum_{i\in\III} L_if_i$ represents $\Geq_0$.
\end{thm}
So, under the stated assumptions, the social observer's beliefs are given by a linear pooling of individual beliefs.

\subsubsection{Convexification}\label{S:convexification}

Here we give an approach that allows the population to be infinite 
and that provides an axiomatic basis for the use of vector  measures.  
It mixes objective and subjective probabilities in the style of \citet{AA1963}.  

We proceed by embedding $X$ in a convex set $\ov X$. 
Here we only assume that $X$ is a Boolean algebra on a set $S$ of states of nature. 
Say that an \emph{extended event} is a function $F\colon S\to[0,1]$ that is constant on each cell of a finite partition $\AA\subset X$ of $S$. We  identify each $A\in X$ with the extended event $\chi_A$ given by the characteristic function of $A$. Let $\ov X$ be the set of extended events; it is a convex subset of the vector space of all functions $S\to \RRR$. 

A preorder $\ov \Geq$ on $\ov X$ can be understood as a comparative likelihood relation in the following way. Suppose that, for each $p \in [0,1]$, a coin with bias $p$ is going to be 
tossed, independently of the events encoded in $X$;
let $H_p$ be the event that it lands heads.%
\footnote{At a cost in abstraction, we could avoid the need for infinitely many coins by considering a single sample from $[0,1]$ with the uniform measure.} 
Given $F \in \ov X$, constant on 
each cell of some partition $\{E_1, \dots, E_n\} \subset X$ of $S$, 
we associate the event $H_F = E_1 H_{p_1} \lor \dots \lor E_n H_{p_n}$, where $p_j$ is the value that $F$ takes on $E_j$. 
Thus $H_F$ occurs whenever nature selects some event $E_j$ from the partition, and, independently, the coin with bias $p_j$ lands heads. 
$F \, {\ov\Geq} \, F'$ holds just in case $H_F$ is  
 judged at least as likely as $H_{F'}$.%
 \footnote{\label{fn:AA DT}While emphasizing that our conceptual approach is not decision theoretic, $\ov \Geq$ should match the preferences of an agent who gets a prize on heads.} 
So interpreted, $\ov \Geq$ should appropriately 
take into account the objective probabilities of the coin tosses;
we suggest that the appropriate constraint is strong independence.

The following proposition allows us to generate representations of likelihood relations on $X$ by vector measures.

\begin{lemma}\label{L:SI2} 
Let $X$ be a Boolean algebra and $\ov{X}$ the set of extended events.
\begin{enumerate}[label=(\roman*), nosep]
\item \label{part:SI2-a} If a preorder  $\ov{\Geq}$ on $\ov{X}$ satisfies strong independence, its restriction $\Geq$ to $X$ can be represented by a vector measure. Conversely, a preorder $\Geq$ on $X$ can be represented by a vector measure only if it arises in this way.
\item \label{part:SI2-b} If a mixture-preserving function $\ov f \colon \ov X \to V$, with $V$ a vector space,  satisfies  
$\ov f(\chi_\varnothing) = 0$, its restriction $f$ to $X$ is a vector measure. Conversely, a function $f\colon X\to V$ is a vector measure only if it arises in this way.
\end{enumerate}
\end{lemma}

Lemma~\ref{L:SI2}, in combination with Lemma~\ref{L:SI}, 
yields the following application of Theorem \ref{T:infinitaryH}.

\begin{thm}\label{T:convexification}
Let $X$ be a Boolean algebra, and $\ov X$ the space of extended events. For $i\in\III\cup\{0\}$, let $\ov \Geq_i$ be a preorder on $\ov X$ whose restriction to $X$ is $\Geq_i$. 
Assume that each $\ov \Geq_i$ satisfies strong independence, and that the $\ov{\Geq}_i$ together satisfy \ref{P1}--\ref{P4}.

Then there is for each $i\in\III\cup\{0\}$ a vector measure $f_i \colon X\to V_i$  that  represents $\Geq_i$, and a strictly positive linear map $L\colon V_\III\to V_0$, with each  $L_i$ an order embedding, such that $f_0=L\fI$. 
\end{thm}

Let us contrast the two routes to linear pooling.
Convexification allows the population to be infinite, and avoids the restriction of the Lyapunov approach to $\Geq_i$ with admissible representations.  
In addition, Lemma~\ref{L:SI2} provides necessary and sufficient conditions for 
a preorder $\Geq$ on $X$ 
to be represented by a vector measure, whereas no such result has been given for $\Geq$ 
to be represented by an admissible vector measure. 
On the other hand, convexification 
requires extending $X$ to $\ov X$, along with Pareto and strong independence for the extended preorders. 
It is hard to see  an objection to this extension of Pareto,%
\footnote{We note, however, that if the $\Geq_i$ satisfy \ref{P1}--\ref{P4}, and are extended to strongly independent preorders $\ov \Geq_i$ on $\ov X$, it does not follow that the $\ov \Geq_i$ satisfy  \ref{P1}--\ref{P4}.}
but strong independence has its critics (even while remaining very popular) in the standard decision theoretic version of the Anscombe-Aumann framework.%
\footnote{See e.g. \citet{iG2009} for discussion.}
Those with similar doubts in our comparative likelihood framework  (cf. note~\ref{fn:AA DT}) might prefer the Lyapunov approach.  Finally, for those who prefer to avoid objective probabilities, we note the possibility of imposing convex structure directly on $X$; compare 
\citet{GMMS2003}.

\section{Related literature}\label{S:literature}
We conclude by relating our results to extant work.

In section \ref{S:povs} we noted that our use of partially ordered vector spaces generalises some other forms of representation.  These include representations with 
values in a non-Archimedean ordered field $F$, often used to model failures of continuity (see note~\ref{fn:fields} for references). Here, $F$ is generally assumed to be an extension of the real numbers, 
and as such is an ordered real vector space 
(cf.~Example~\ref{ex:DR infinity}). 
However, advocates of this approach may prefer to rework our results using $F$, rather than $\RRR$, as the basic field. This would mean interpreting such conditions as linearity, convexity, and strong independence using coefficients drawn from $F$.%
\footnote{Thus, for example,  a subset $X$ of an $F$-vector space is convex over $F$ if and only if, for all $x,y\in X$ and all $\aa\in F$ with $0<\aa<1$, we have $\aa x + (1-\aa)y\in X$.} 
We do not pursue this project, but the only results that we do not expect to extend almost verbatim are 
Theorems~\ref{T:expectH}  and  \ref{T:pooling-c}.

As already noted, in the special case in which $\III$ is finite and $V_i = \RRR$ for each $i \in \III \cup \{0\}$, Theorem~\ref{T:infinitaryH-2} is  proved in \citet[Prop. 1]{DMM1995}, without requiring \ref{DR}; note that \ref{P4} is vacuous under these assumptions. 
This is the first result we know of that emphasizes the usefulness of co-convexity in the context of social aggregation. They apply their result to 
preference aggregation to obtain Harsanyi's theorem, and to 
opinion pooling using Lyapunov, in the special case in which each $\Geq_i$ is represented by a standard probability measure.

\citet{lZ1997} allows $\III$ to be infinite, and considers the special case in which $X$ is a mixture set 
and the $f_i$ are mixture-preserving representations with values in $\RRR$ (i.e. von Neumann-Morgenstern expected 
utility functions),%
\footnote{As \citet{pM2001} proves, each mixture set $X$ maps naturally onto to a convex set $\overline X$ such that any mixture-preserving function on $X$ comes from one on $\overline X$; thus the use of mixture-preserving functions on mixture sets rather than on convex sets does not give any greater generality.}
so that the $\Geq_i$ satisfy strong independence, continuity and completeness. 
Under these assumptions, his Theorems 1 and 2 are essentially our Theorem~\ref{T:infinitaryH-2}\ref{part:infinitaryH-2 P1} and \ref{part:infinitaryH-2 P2} respectively.%
\footnote{The exact statement of Zhou's Theorem 2 is closer to the variant of Theorem~\ref{T:infinitaryH-2}\ref{part:infinitaryH-2 P2}  outlined in section~\ref{S:domain}.} 
He also raises the question of whether, in our terminology, \ref{DR} can be dropped from Theorem~\ref{T:infinitaryH-2}\ref{part:infinitaryH-2 P2}, and gives an affirmative answer under yet further assumptions.  Our examples in section \ref{S:domain} show that the answer is negative in general. But we emphasise that 
our 
Theorem \ref{T:infinitaryH} and the main part of Theorem \ref{T:infinitaryH-3} do not assume \ref{DR}.

Despite the importance of the problem, the results of \citet{DGT2015} are the only ones we know of that 
generalize Harsanyi's theorem by dropping completeness. 
They assume a finite population, and take $X$ to be the set of probability measures on a finite set of outcomes. 
They assume that the $\Geq_i$ satisfy strong independence and a slight strengthening of mixture continuity. 
It follows that the $\Geq_i$ have expected multi-utility representations \citep{SB1998, DMO2004}, and they present their results in such terms.  
Recall from Example \ref{ex:multi} than any multi-representation can be re-interpreted as a representation in a partially ordered vector space. Our Theorem \ref{T:infinitaryH-2}\ref{part:infinitaryH-2 P1} and \ref{part:infinitaryH-2 P2} say that, given \ref{P1} or \ref{P1} and \ref{P2}, the vector-valued representation of $\Geq_0$ is linearly related to those of the other $\Geq_i$. Danan {\em et al} prove essentially this same result in their setting, and go slightly further in describing the components of the linear relation in terms of the original multi-representations. 
They also explain how these results extend to the case of an infinite population,  
generalizing Zhou's results to allow for incompleteness (although, unlike Zhou, assuming a finite set of outcomes).

Our results differ from theirs in three main ways. First, we do not require continuity; 
one motivation for this is that in the absence of completeness, continuity is a more problematic assumption than it appears.%
\footnote{See \citet{jD2011} and \citet{MM2018}, extending a classic observation of \citet{dS1971}.} 
Second, as emphasized
in section~\ref{S:harsanyi}, 
we allow for a much wider range of 
interpretations of the domain $X$. 
Third, while they only consider the axioms \ref{P1}--\ref{P2}, we consider the effects of adding the standard strong Pareto-style axiom \ref{P3}, and in addition, 
the apparently novel 
\ref{P4}. 
This whole package of Pareto-style axioms is plausible, in our view, and essential to those of our results that are arguably the fullest generalizations of Harsanyi. 
Of course, our results do not wholly replace those of \citet{DGT2015}, as they study details that arise specifically in the context of continuity and multi-utility representations. 
But we believe that it is important to consider the full range of Pareto-style conditions, even in that context;  our results provide a framework for doing so, although we do not pursue it here.

As in Harsanyi's theorem, the results of \cite{lZ1997}, \cite{DGT2015} and this paper do not assume interpersonal comparisons.  
Immediately after the statement of his main theorem, though, \citet{jH1955} introduces interpersonal comparisons and a form of anonymity, leading to the further conclusion that 
the social preorder can be represented by the sum of real-valued mixture-preserving individual utility functions that have been normalized to reflect interpersonal comparisons. Generalizations of this version of Harsanyi's result, that assume interpersonal comparisons and anonymity at the outset but drop one or more of continuity, completeness, strong independence, and the requirement of a finite population, are given in \citet{mF2009}, \citet{mP2013, mP2014}, and \citet{MMT2018}.%
\footnote{While it would be straightforward to add interpersonal comparisons and anonymity to the results of this paper when $\III$ is finite, there is a well known incompatibility between full anonymity and Pareto when $\III$ is infinite. This is avoided in \citet{mP2014} by restricting to finite anonymity, and in \citet{MMT2018} 
by restricting to social lotteries in which only finitely many individuals have a chance of existing.}

Generalizations of Harsanyi to the case where preferences are over Anscombe-Aumann or Savage acts, allowing individuals to have different beliefs and values, are well known to lead to impossibility results.%
\footnote{From among a very large body of literature, see e.g. \citet{jB1990, pM1995, pM1998, MP2015}; and \citet{sZ2016}.} 
Our treatment of opinion pooling separately from preference aggregation may therefore be seen as following the advice of, for example, \citet[pp. 352--3]{pM1998}  
to aggregate opinion prior to aggregating preferences.

The classic axiomatization of linear opinion pooling, when the population is finite and the $p_i$ are standard ($[0,1]$-valued) probability measures on a $\sigma$-algebra, was given in \citet{kM1981}. 
An axiomatization using analogues of the Pareto-style conditions \ref{P1}--\ref{P3} 
was given in \citet{pM1995}; see also \citet{DMM1995}. 
\citet{cC2007} extends that approach to axiomatize linear pooling for ordinal probabilities. 
Extending McConway's approach, linear pooling results for finitely additive probability measures and an infinite population are given in \citet{fH2015} and \citet{mN2019}; 
the latter giving an equivalent  result using analogues of 
\ref{P1} and \ref{P2}.

\section{Proofs}
We begin with some simple 
observations about convex sets, cones and vector preorders.

\begin{lemma}\label{L:Vspan} 
Let $Z$ be a nonempty convex set. Then $\Span(Z-Z)=\Set{\ll(z-z')}{\ll\in\RRR,\ll\ > 0,z,z'\in Z}.$ 
\end{lemma}
\begin{proof}
The right-hand side is clearly contained in and spans the left-hand side. It suffices to show that it equals its own span. It is clearly closed under scalar multiplication. To show that it is closed under addition, suppose given $\lambda(z-z')$ and $\mu(w-w')$ with $\ll,\mu>0$ and $z,z',w,w'\in Z$.   Then it is easy to check
$\lambda(z-z')+\mu(w-w')=\nu(v-v')$,
for $\nu\coloneqq\mu+\nu$, 
$v\coloneqq\frac{\ll}{\nu} z +\frac{\mu}\nu w$, and
$v'\coloneqq\frac{\ll}{\nu} z' +\frac{\mu}\nu w'$. 
Note that $\nu>0$ and that $v,v'$ are elements of $Z$, since it is convex. 
\end{proof}

Given a vector space $V$, recall that $C\subset V$ is a {\em convex cone} if  $0\in C$, 
$C+C\subset C$ and $\lambda C\subset C$ for all $\lambda>0$. 
Clearly if $C, C' \subset V$ are convex cones, then so is $C+C'$. 
The following is well-known; see e.g. \citet[G.1.3]{eO2007}.

\begin{lemma}\label{l:ccV+0}
Let $C\subset V$, where $V$ is a vector space. The binary relation $\Geq_V$ on $V$ defined by 
 $v\Geq_V w \IFF v-w\in C$ is a linear preorder if and only if $C$ is a convex cone. 
Conversely, any linear preorder on $V$ is of this form. 
\end{lemma}

\begin{lemma}\label{L:convex cone} 
Suppose $Z$ is a nonempty subset of a vector space $V$, and $\Geq$ is a preorder on $Z$ that satisfies strong independence. 
Let $C \coloneqq \Set{\lambda(x-y)}{\lambda\in \RRR, \lambda> 0, x,y\in Z, x\Geq y}$. Then $C$ is a convex cone in $V$. 
Equip $V$ with the linear preorder $\GeqV$ defined by 
$v\GeqV w \iff v-w\in C$. 
Then the inclusion $\iota \colon Z \to V$ represents $\Geq$. 
\end{lemma}

\begin{proof}
It is clear that $0 \in C$, and that for $\lambda >0$, $\lambda C \subset C$. To show that $C + C \subset C$, let $v, w \in C$. Then $v = \lambda(x-y)$ and $w = \mu(x'-y')$ 
for some $\lambda, \mu >0$, $x, x', y, y' \in Z$ with $x \Geq y$ and $x' \Geq y'$. 
We have
\[
	v+w = (\lambda + \mu)\Big[\big(\tfrac{\lambda}{\lambda + \mu}x + \tfrac{\mu}{\lambda + \mu}x'\big) - \big(\tfrac{\lambda}{\lambda + \mu}y + \tfrac{\mu}{\lambda + \mu}y'\big)\Big].
\] 
Let $\aa = \tfrac{\lambda}{\lambda + \mu}$. Since $\Geq$ satisfies strong independence, we have $\aa x + (1-\aa)x' \Geq \aa y + (1-\aa)x'$ and $\aa y + (1-\aa)x' \Geq \aa y + (1-\aa)y'$, and thus $\aa x + (1-\aa)x' \Geq \aa y + (1-\aa)y'$. The displayed equation then shows that $v + w \in C$, 
establishing that $C$ is a convex cone.

We now show that for $x, y \in Z$, $x \Geq y \IFF x \Geq_V y$. Clearly $x \Geq y \Rightarrow x - y \in C \Rightarrow x \Geq_V y$. Conversely, suppose  $x \Geq_V y$. Then $x - y \in C$, hence there exist $\lambda >0$, $x', y' \in Z$ with $x' \Geq y'$ and $x - y = \lambda(x'-y')$. 
Letting $\aa = \tfrac1{1+\ll}$, this rearranges to  $\aa x + (1-\aa) y' = \aa y+ (1-\aa) x'$.    
Since $x'\Geq y'$ and $\Geq$ satisfies strong independence, we have $\aa x + (1-\aa)x' \Geq \aa x + (1-\aa) y'=\aa y + (1-\aa)x'$. By strong independence again, we must have $x\Geq y$. 
\end{proof}

Recall that we write $1_i \colon V_i \to V_\III$ for the natural embedding of $V_i$ into $V_\III$.

\begin{proof}[Proof of Theorem~\ref{T:infinitaryH}]
First let us verify that, in each case, the right-hand side entails the left-hand side. So suppose  we have a partially ordered vector space $(V, \GeqV)$ and  a linear $L\colon V_\III\to V$ such that $L\fI$ represents  $\Geq_0$.
Recall that we equip $V_\III$ with the product partial order $\GeqP$.

Suppose first that $x\Eq_i y$ for every $i\in\III$. 
Then $\fI(x)=\fI(y)$, hence $L\fI(x) = L\fI(y)$. 
Since $L\fI$ represents $\Geq_0$, we find that $x\Eq_0 y$. Therefore \ref{P1} holds. 
Next, suppose that $L$ is positive, and suppose $x\Geq_i y$ for every $i\in\III$. 
Then $\fI(x) \GeqP \fI(y)$, hence $L\fI(x) \GeqV L\fI(y)$, so $x\Geq_0 y$. Thus \ref{P2} holds as well as \ref{P1}. 
Similarly, suppose that $L$ is strictly positive. If $x\Geq_i y$ for all $i\in\III$ and $x\sGeq_jy$ for some $j\in\III$, then $\fI(x)\sGeqP \fI(y)$, whence $L\fI(x)\sGeqV L\fI(y)$ and $x\sGeq_0 y$; thus \ref{P3} holds.  This covers  the right-to-left directions in \ref{part:infinitaryH P1}--\ref{part:infinitaryH P3}.

As for \ref{part:infinitaryH P4}, 
suppose that $L$ is strictly positive 
and each $L_i$ is an order embedding. 
Suppose that 
$x \Geq_i y$ for all $i \in \III \setminus \{j\}$ and $x \Incom_j y$. 
 This implies that $f_i(x) \Geq_{V_i} f_i(y)$ for all $i \in \III \setminus \{j\}$ and 
$f_j(x) \Incom_{V_j} f_j(y)$.  
We can therefore write $\fI(x)-\fI(y)=v_1+v_2$ with $v_1\GeqP 0$ and $v_2=1_j(f_j(x)-f_j(y))\IncomP 0$. 
 Since 
$L$ is strictly positive, 
 we have $L(v_1)\GeqV 0$; 
 since $L_j$ represents $\Geq_{V_j}$, 
 we have $L(v_2)\IncomV 0$. Therefore $L(\fI(x)-\fI(y))=L(v_1)+L(v_2) \not\LeqV 0$. 
 This implies 
 $x\not\Leq_0 y$. 
 This shows that \ref{P4} must hold as well as 
 \ref{P1}\!--\ref{P3}. 

Conversely, now, we show that the left-hand side entails the right in each case.

Assume \ref{P1}.
Define subsets of $V_\III$:
\[
\begin{aligned}
C_0 &=\Set{\lambda(\fI(x)-\fI(y))}{\lambda> 0, x, y \in X, x\Geq_0 y}\\
C_P &=\Set{v}{v\in V_\III, v \GeqP 0}\\
C &=C_0 + C_P.
\end{aligned}
\]
We first prove that these are convex cones. Let $X'=(\fI,f_0)(X)\subset V_\III\times V_0$. 
By the co-convexity assumption, $X'$ is convex.  
Let $\pi_\III$ and $\pi_0$ be the projections of $V_\III\times V_0$ onto $V_\III$ and $V_0$ respectively. 
Let $\Geq'$ be the preorder on $X'$ represented by the restriction of $\pi_0$ to 
$X'$: $x\Geq'y \iff \pi_0(x')\Geq_{V_0} \pi_0(y')$. This $\Geq'$ satisfies strong independence, as in Lemma~\ref{L:SI}\ref{part:SIa}. 
Define 
a subset of $V_\III\times V_0$: 
\[
C'_0 = \Set{\lambda(x'-y')}{\lambda> 0, x',y'\in X', x'\Geq' y'}.
\]
By Lemma~\ref{L:convex cone}, $C'_0$ is a convex cone. 
Note that, for $x'=(\fI,f_0)(x)$ and $y'=(\fI,f_0)(y)\in X'$, we have 
\begin{equation}\label{eq:Geq'}
x'\Geq' y'\iff \pi_0(x')\Geq_{V_0} \pi_0(y') \iff f_0(x)\Geq_{V_0} f_0(y)\iff x\Geq_0 y.
\end{equation}
This shows that $C_0 = \pi_\III(C'_0)$, and therefore that $C_0$ is a convex cone. 
Meanwhile, $C_P$ is a convex cone by Lemma \ref{l:ccV+0}.
Therefore the sum $C$ is also a convex cone. 

Next we prove 
\begin{equation}\label{eq:equivalence2}
     x\Geq_0 y\iff \fI(x)-\fI(y)\in C_0. 
\end{equation}
The left-to-right direction is obvious
from the definition of $C_0$. For the converse, we first show that $\pi_\III$ is injective on $Y \coloneqq \Span(X'-X')$. 
Since $\pi_\III$ is linear and $Y$ is a vector space, it suffices to show that if $\pi_\III(v)=0$ for some $v\in Y$, then $v=0$. So suppose $\pi_\III(v)=0$. 
By Lemma~\ref{L:Vspan}, we can write $v = \ll(a'-b')$, with $a',b' \in X'$, $\ll >0$;  by definition of $X'$, we can write 
$a'=(\fI(a),f_0(a))$  and
$b'=(\fI(b),f_0(b))$ with $a,b\in X$. 
Then $\pi_\III(v)=0$ implies $\fI(a)=\fI(b)$;  so, by \ref{P1}, $a \Eq_0 b$,  and hence $f_0(a) = f_0(b)$.  Therefore $a'=b'$ and $v=0$,  establishing injectivity of $\pi_\III$ on $Y$.

Now suppose $\fI(x)-\fI(y)\in C_0$. Let $x'=(\fI,f_0)(x)$ and $y'=(\fI,f_0)(y)$. 
Note that $\pi_\III(x' -y') = \fI(x) -\fI(y) \in C_0$. 
Since $x'-y'\in Y$, $C_0 = \pi_\III(C'_0)$, $C'_0 \subset Y$, and $\pi_\III$ is injective on $Y$, 
we can conclude that $x'-y'\in C'_0$. 
By Lemma \ref{L:convex cone}, however, $x'-y'\in C'_0$ implies $x'\Geq' y'$, and by  \eqref{eq:Geq'} we obtain $x \Geq_0 y$. 
This establishes the right-to-left direction in \eqref{eq:equivalence2}.

With these preliminaries, we now prove the left-to-right direction in part \ref{part:infinitaryH P1}. Assuming \ref{P1}, 
we define a partially ordered vector space $(V, \GeqV)$ as follows. 
We first let $V$ be the quotient of $V_\III$ by the subspace $C_0\cap-C_0$, and let $L\colon V_\III\to V$ be the quotient map. 
$L(C_0)$ is a convex cone in $V$, so it defines a linear preorder $\GeqV$ on $V$, by Lemma \ref{l:ccV+0}. Namely, we have
\[
L(v)\GeqV L(w) \iff   
L(v-w)\in L(C_0)
\iff v-w\in C_0.
\]
It follows that $\GeqV$ is a partial order: if $L(v)\EqV L(w)$ then $v-w\in C_0\cap -C_0$, so $L(v)=L(w)$. We claim that $L\colon V_\III\to V$ is the required map for part~\ref{part:infinitaryH P1}: 
that is, we claim that $L\fI$ represents $\Geq_0$. Suppose first that $x\Geq_0 y$. Then 
$\fI(x)-\fI(y)\in C_0$, 
so $L\fI(x)\GeqV L\fI(y)$. Conversely, if $L\fI(x)\GeqV L\fI(y)$, then 
$\fI(x)-\fI(y)\in C_0$.
Therefore, using \eqref{eq:equivalence2}, we find 
$x\Geq_0 y$, as desired.

For part \ref{part:infinitaryH P2}, further assuming \ref{P2},
we instead define $V$ to be the quotient of $V_\III$ by $C\cap -C$. We let $L\colon V_\III\to V$ be the quotient map, and now we equip $V$ with the linear partial order $\GeqV$ defined by 
\begin{equation}\label{eq:L for P2}
L(v)\GeqV L(w) \iff   
L(v-w)\in L(C) \iff
v-w\in C.
\end{equation}
It is clear from the fact that $C$ contains $\CP$ that 
$L$ is positive. 
It remains to prove 
that $L\fI$  represents $\Geq_0$. Suppose first that $x\Geq_0 y$. Then 
$\fI(x)-\fI(y)\in C_0\subset C$, 
so $L\fI(x)\GeqV L\fI(y)$. 
Conversely, if $L\fI(x)\GeqV L\fI(y)$, then 
$\fI(x)-\fI(y)\in C$. 
We may therefore write  
$\fI(x)-\fI(y)= v_0 + v_P$ for some $v_0 \in C_0, v_P \in \CP$. 
Solving this equation for $v_P$,  
we find $v_P \in \Span(\fI(X)-\fI(X))$. 
Since $\fI(X)$ is convex,  
Lemma~\ref{L:Vspan} implies 
$v_P = \lambda(\fI(x') - \fI(y'))$ for some $\lambda > 0$, $x', y' \in X$. 
Since $v_P \in \CP$, 
\ref{P2} 
implies $x' \Geq_0 y'$, hence, 
by \eqref{eq:equivalence2},
$v_P \in C_0$. 
This implies $\fI(x) - \fI(y) \in C_0$, 
and hence by \eqref{eq:equivalence2} again, 
$x \Geq_0 y$.

For parts
\ref{part:infinitaryH P3} and \ref{part:infinitaryH P4}, we use the same $V$, $L$, and $\GeqV$ as in part \ref{part:infinitaryH P2}. For part \ref{part:infinitaryH P3}, it only remains to show that if \ref{P1}--\ref{P3} hold,
 then $L$ is strictly positive. 
As before, $L$ is positive, and 
$\GeqV$ is a partial order; because of this,
it suffices to show that, 
for $v\in V_\III$,  $v\sGeq_{P} 0$ rules out $Lv=0$.  
Suppose on the contrary that $Lv=0$.  Then, by the way $L$ was defined, we must have $v\in C\cap -C$. We can therefore write
$-v = v_0 + v_P$ with $v_0 \in C_0$ and $v_P \in \CP$. 
Since $v_0\in C_0$, we can further write $v_0 = \lambda (\fI(x) - \fI(y))$ with $\lambda >0$, $x \Geq_0 y$. Rearranging, we find $\fI(y)-\fI(x)=\tfrac1\lambda(v +v_P)$. So if $v \sGeq_{P} 0$, we have $\fI(y)-\fI(x) \sGeqP 0$. 
\ref{P3} yields $y \sGeq_0 x$, a contradiction.

Finally, for part \ref{part:infinitaryH P4}, it suffices to show that, if \ref{P1}--\ref{P4} hold, then every $L_i$ is an order embedding: 
$v \GeqVi 0 \IFF L_i(v) \GeqV 0$.  
Since, as just established, $L$ is strictly positive, so is each $L_i$, and it  remains to show that if $L_i(v)\GeqV 0$, then $v\Geq_{V_i} 0$.
Suppose therefore that $L_i(v) \GeqV 0$. 
That is, $L1_iv\GeqV 0$, so 
by \eqref{eq:L for P2}, 
$1_i v\in C$ and  we can write 
$1_i v = v_0 + v_P$ for some $v_0 \in C_0$, $v_P \in C_P$.
We may further write $v_0 =  \lambda(\fI(x)-\fI(y))$ 
for some $\lambda > 0$, $x \Geq_0 y$. 
Rearranging, we find $\fI(y)-\fI(x)=\tfrac1\lambda(-1_iv +v_P).$ 
Therefore, $f_j(y)\GeqVj f_j(x)$  for all $j\in\III\setminus\{i\}$.
By \ref{P3} and \ref{P4} respectively, we will have 
$y \nLeq_0 x$ 
(a contradiction) if 
$f_i(y)\sGeqVi f_i(x)$ or $f_i(y)\IncomVi f_i(x)$; therefore we must have $f_i(x)\GeqVi f_i(y)$.
Now, by choice of $v_0$ and $v_P$,  $v=\lambda(f_i(x)-f_i(y))+(v_P)_i$; both terms on the right are $\GeqVi 0$, so we find $v\GeqVi 0$, as desired.
\end{proof}

The proof of Theorem~\ref{T:infinitaryH-2} rests on the following 
`abstract Harsanyi theorem'.

\begin{thm}\label{t:AbsHarsanyi} 
Let $X$ be a nonempty set. Let $Y$ and $Z$ be vector spaces and let $f \colon X \to Y$, $g \colon X \to Z$ 
be co-convex. 
 Then 
\begin{equation}\label{eq:inj}
     g(x)=g(x') \THEN f(x)=f(x') \textup{ for all } x,x'\in X \end{equation}
if and only if $f=Lg+y_0$ for some linear $L \colon Z\to Y$
and some $y_0 \in Y$. 
Moreover, the restriction of $L$ to $\Span(g(X)-g(X))$ is uniquely determined. 
\end{thm}

\begin{proof}
To take the last statement first, suppose that $f=Lg+y_0$ and also $f=L'g+y'_0$. Subtracting, we see that $L$ and $L'$ differ by the constant $y'_0-y_0$ on $g(\X)$, and therefore they are equal on $\Span(g(\X)-g(\X))$.

For the first statement, 
it is clear that \eqref{eq:inj} holds 
if $f$ is of the form $f=Lg+y_0$. For the converse, 
let $X' \coloneqq (f,g)(X)$; it is a convex set by assumption. 
Let $A \coloneqq \Span(X'-X')$; it is a linear subspace of $Y \times Z$. 
In light of Lemma~\ref{L:Vspan} applied to $X'$, the condition \eqref{eq:inj}
is equivalent to the condition that $A$ contains no elements of the form $(y,0)$ with $y \neq 0$. 
Since $A$ is a linear subspace, we find that 
\[
	(y,z), (y',z) \in A \implies y=y'.
\]
$A$ is therefore the graph of a partial function $L$ from $Z$ to $Y$. 
By definition, the domain of $L$ is the projection of $A$ to $Z$, namely $\Span(g(X)-g(X))$, 
and $L$ is characterized by the equation 
\[
	A= \{(L(z), z) \colon z \in \Span(g(X)-g(X))\}.
\]
Also, $L$ is a linear function since $A$ is a linear subspace. 
Extend $L$ arbitrarily to a linear function from $Z$ to $Y$. 
Fix $(y,z) \in X'$ and set $y_0 = y-L(z)$. 	
Then for any $x \in X$, 
we have $f(x) = L(g(x)) + y_0$. 
\end{proof}

\begin{proof}[Proof of Theorem~\ref{T:infinitaryH-2}]
Note that, if $L\fI+b$ represents $\Geq_0$, then so does $L\fI$. Thus the right-to-left direction in each part is a special case of the corresponding claim in Theorem~\ref{T:infinitaryH}.

For the left-to-right directions, assume for the remainder that \ref{P1} holds. 
This implies $\fI(x) = \fI(x') \THEN f_0(x) = f_0(x')$ for $x, x' \in X$. Co-convexity of the $f_i$ is equivalent to co-convexity of $\fI$ and $f_0$. 
Theorem \ref{t:AbsHarsanyi} 
therefore 
yields a linear map $L\colon V_\III\to V_0$ and some $b\in V_0$ such that $f_0=L\fI+b$, establishing  part~\ref{part:infinitaryH-2 P1}. 

Let $v \in V_\III$. 
By the co-convexity assumption, 
$\fI(X)$ is convex, so \ref{DR} and Lemma~\ref{L:Vspan} imply that 
there exist $x, y \in X$, $\lambda >0$ such that $\fI(x) - \fI(y) = \lambda v$. 
Since $L\fI$ represents $\Geq_0$, 
\[
	x \Geq_0 y \iff Lv \Geq_{V_0} 0.
\]
We have $v \GeqP 0 \IFF x \Geq_i y$ for every $i \in \III$. 
If \ref{P2} holds, $v \GeqP 0 \Rightarrow x \Geq_0 y$; 
if \ref{P3} holds $v \sGeqP 0 \Rightarrow x \sGeq_0 y$. 
Using the displayed equivalence, if \ref{P2} holds, $L$ is positive, and if \ref{P2}\!--\ref{P3} hold, $L$ is strictly positive. 
This establishes \ref{part:infinitaryH-2 P2}--\ref{part:infinitaryH-2 P3}. 

Now assume \ref{P1}--\ref{P4} hold. By Theorem~\ref{T:infinitaryH}\ref{part:infinitaryH P4}, 
there exists a partially ordered vector space $(V', \Geq_{V'})$ and a linear map $L' \colon V_\III \to V'$ such that 
$L'\fI$ represents $\Geq_0$,
with every $L'_i$ an order embedding. Let $v \in V_i$. By convexity of $\fI(X)$, \ref{DR} 
and Lemma~\ref{L:Vspan} again, there exist $x, y \in X$, $\lambda >0$ such that $\fI(x) - \fI(y) = \lambda 1_iv$. 
Then $v \Geq_{V_i}0 \iff L'_i v \Geq_{V'} 0 \iff L'\fI(x) \Geq_{V'} L'\fI(y)
\iff x \Geq_0 y \iff f_0(x) \Geq_{V_0} f_0(y)
\iff L\fI(x) \Geq_{V_0} L\fI(y) \iff L_i v \Geq_{V_0} 0$.
This shows that $L_i$ represents $\Geq_{V_i}$. 
Since $V_0$ and $V_i$ are partially ordered,  it also shows that $\ker L_i = \{0\}$: $L_iv=0\iff L_iv\Eq_{V_0}0\iff v\Eq_{V_i} 0\iff v=0$. Therefore $L_i$ is injective, and 
$L_i$ is an order embedding, establishing \ref{part:infinitaryH-2 P4}. 
\end{proof}

\begin{proof}[Proof of Theorem~\ref{T:infinitaryH-3}]
For the right-to-left direction, suppose we are given $V$, $g_i$, and $S$ satisfying the three conditions \ref{part:H-3a}--\ref{part:H-3c}, or even just \ref{part:H-3a} and \ref{part:H-3b}.
Suppose $x, y \in X$ with $x \Geq_i y$ for all $i\in\III$.
Since $g_i$ represents $\Geq_i$, this implies $\gI(x)-\gI(y) \GeqP 0$, where now $\GeqP$ is the product partial order on $\ProdI V$. Since 
$\gI(x)-\gI(y)\in Y\coloneqq  \ProdI \Span(g_i(X)-g_i(X))$,
and $S$ is assumed strictly positive on $Y$, 
we find 
$S\gI(x) \Geq_{V} S\gI(y)$, and since 
$S\gI$ represents $\Geq_0$, this implies $x \Geq_0 y$. 
This establishes \ref{P1} and \ref{P2}. If $x\sGeq_i y$ for some $i\in\III$, then a similar argument yields $x\sGeq_0 y$, and hence \ref{P3}. 

For \ref{P4}, suppose $x \Geq_i y$ for all $i \in \III\setminus\{j\}$ and $x \Incom_j y$. Let $1_j$ be the inclusion of $V$ into $\ProdI V$ as the $j$th factor.
Set $v \coloneqq g_j(x)-g_j(y)$, so that $v\Incom_V 0$.
We may write $\gI(x) - \gI(y) = v_P + 1_jv$ with $v_P \GeqP 0$, and $v_P,1_jv\in Y$. Since $S$ is strictly positive on $Y$,  we have $Sv_P\Geq_V 0$, 
and since $S$ extends summation we have $S1_jv=v\Incom_V 0$. Therefore (since $\Geq_V$ is a linear preorder) we have $Sv_P+S1_jv\nLeq_V 0$, and since $S\gI$ represents $\Geq_0$ we find $x\nLeq_0 y$, as desired.

Conversely, suppose that \ref{P1}--\ref{P4} hold. 
From Theorem~\ref{T:infinitaryH}\ref{part:infinitaryH P4}, 
there is a strictly positive linear map  $L\colon V_\III\to V$ with values in some partially ordered vector space $V$, such that $L\fI$ represents $\Geq_0$ and every $L_i$ is an order embedding.%
\footnote{Anticipating the proof of Theorem~\ref{T:expectH}, we note that the following construction of the $g_i$ and $S$ does not depend on the $V_i$ or $V$ being partially ordered, rather than merely preordered.}
This implies that  $g_i \coloneqq L_i f_i$ is a representation of of $\Geq_i$ with values in  $V$ 
for each $i \in \III$.  
Together the $g_i$ define $\gI\colon X\to \ProdI V$. 

We now define the map $S\colon\ProdI V\to V$. 
 The construction is illustrated by the following commutative diagram.
\[
\begin{tikzcd}
X \arrow[d,"\fI"]\arrow[drr,"\gI"] \\
 \ProdI V_i  \arrow[r, "(L_i)"] \arrow[dr, "L"] & \ProdI V'_i \arrow[d, "S_1"]  \arrow[r,hook] & \ProdI V  \arrow[dl,"S"]\\
 & V & \arrow[l, "S_2"] \SumI V \arrow[u,hook]
\end{tikzcd}
\]
 Here, let $V'_i \coloneqq L_i(V_i)$; it is a subspace of $V$ order-isomorphic to $V_i$. 
(Note that $L_i$  is injective, since it is an order embedding.) 
The $L_i$ together define an isomorphism $(L_i)\colon \ProdI V_i \to \ProdI V'_i$. 
On the other hand, $L$ is a map from $V_\III=\ProdI V_i$ to $V$. 
We therefore obtain a map $S_1\colon \ProdI V'_i\to V$ by 
$S_1=L\circ(L_i)\inv$.

Next, we have the summation $S_2\colon \SumI V\to V$. The domains of $S_1$ and $S_2$ are both subspaces of $\ProdI V$. The intersection of their domains is $\SumI V'_i$. 
Given 
$v\in V_i'$, 
$L(1_iL_i\inv(v))=v$
(where here $1_i$ is the inclusion of $V_i$ into $\ProdI V_i$).
This shows that $S_1$ and $S_2$ coincide on each $V'_i$, and therefore on $\SumI V'_i$. 
 Therefore there exists a linear map $S \colon \ProdI V\to V$ that extends both $S_1$ and $S_2$.

By construction, $\gI=(L_i)\fI$ and therefore
$S\gI = L\fI$. Therefore $g_0\coloneqq S\gI$ represents $\Geq_0$, as required for part \ref{part:H-3a}.  

For part \ref{part:H-3b} it remains to show that $S$ is strictly positive on $Y$. 
Note that in the commutative diagram  defining $S$, the horizontal map $(L_i)\colon \ProdI V_i \to \ProdI V$ is an order embedding, under which $Y$ is the image of $Y'\coloneqq \ProdI \Span(f_i(X)-f_i(X))$. Therefore $S$ is strictly positive on $Y$ if and only if $L$ is strictly positive on $Y'$. 
But $L$ is strictly positive  on $V_\III$, which contains $Y'$.

Part \ref{part:H-3c} is straightforward from the fact that $\fI$ and $f_0$ are co-convex and $\gI$ and $g_0$ are linear transforms of $\fI$.   

For the last claim of the theorem, given \ref{DR} as well as \ref{P1}--\ref{P4}, 
Theorem 
\ref{T:infinitaryH-2}\ref{part:infinitaryH-2 P4} 
provides a linear $L \colon V_\III \to V_0$ such that $f_0 = L\fI +b$ for some $b\in V_0$. 
Repeating the construction above, now with $V=V_0$, we 
obtain $g_i$ and $S$ satisfying conditions \ref{part:H-3b} and \ref{part:H-3c} but instead of \ref{part:H-3a} we have $f_0=L\fI+b=S\gI+b$. 
However, picking any $i\in\III$, we can replace $g_i$ by $g_i-b$ to obtain $f_0=S\gI$.
\end{proof}

\begin{proof}[Proof of Lemma~\ref{L:coconvex0}] 
We have $f(x)=f(x')\iff x\Eq x'\iff g(x)=g(x')$. 
Since $f$ and $g$ are pervasive, 
Theorem~\ref{t:AbsHarsanyi} tells us that there is a unique linear $L \colon V \to V'$ and $b \in V'$ such that 
$g = Lf +b$, and also a unique linear $L' \colon V' \to V$ and $b' \in V$ such that 
$f = L'g +b'$. 
(Note that here $b$ and $b'$ are uniquely determined as well as $L$ and $L'$.)
Together we find $f = L'Lf + (L'b+b')$. 
We also have $f = \id_V f$, so the uniqueness statement in Theorem~\ref{t:AbsHarsanyi} implies that $L'L =  \id_V$. 
A similar argument gives $LL' =  \id_{V'}$, showing that $L$ is bijective.
It remains to show that $L$ is an order isomorphism:   for all $v_1,v_2\in V$, we want 
$v_1\Geq_V v_2\iff Lv_1\Geq_{V'} Lv_2$. By Lemma \ref{L:Vspan}, and the fact that $f$ is pervasive, there exist $x_1,x_2\in X$ and $\lambda>0$ such that $v_1-v_2=\lambda(f(x_1)-f(x_2))$. We have 
$v_1\Geq_V v_2 
\iff f(x_1)\Geq_V f(x_2) \iff x_1\Geq x_2 \iff g(x_1)\Geq_{V'} g(x_2) \iff Lf(x_1)\Geq_{V'} Lf(x_2)\iff L(v_1)\Geq_{V'} L(v_2)$, as desired. 
\end{proof}

\begin{proof}[Proof of Proposition \ref{prop:uniqueness1}]
$\Span(L\fI(X)-L\fI(X))=L\Span(\fI(X)-\fI(X))$ by the linearity of $L$, and $\Span(\fI(X)-\fI(X))=V_\III$, by \ref{DR}. Thus $\Span(L\fI(X)-L\fI(X))=L(V_\III)$, showing that $L\fI\colon V_\III \to L(V_\III)$ is pervasive. So too is $L'\fI\colon V_\III\to L'(V_\III)$.

We next show that $L\fI$ and $L'\fI$ are co-convex. Fix any $x_1,x_2\in X$ and 
$\alpha\in[0,1]$. Since $\fI(X)$ is convex, there exists $x_3\in X$ such that $\fI(x_3)=\alpha \fI(x_1)+(1-\alpha )\fI(x_2)$. Therefore 
$L\fI(x_3)=\alpha L\fI(x_1)+(1-\alpha )L\fI(x_2)$, and similarly 
$L'\fI(x_3)=\alpha L'\fI(x_1)+(1-\alpha )L'\fI(x_2)$. Therefore
$(L\fI,L'\fI)(x_3)=\alpha (L\fI,L'\fI)(x_1)+(1-\alpha )(L\fI,L'\fI)(x_2)$. Thus $L\fI$ and $L'\fI$ are co-convex.

By Lemma~\ref{L:coconvex0} there is therefore 
a unique linear order isomorphism $M\colon L(V_\III)\to L'(V_\III)$ and $b \in L'(V_\III)$ such that $L'\fI=ML\fI +b$. 
Suppose given $v\in V_\III$; by \ref{DR} and Lemma \ref{L:Vspan} there are $x_1,x_2\in X$, and $\lambda>0$, such that $v=\lambda(\fI(x_1)-\fI(x_2))$.  From the fact that 
 $L'\fI=ML\fI +b$ 
it follows that $L'v=MLv$, as required.
\end{proof}

\begin{proof}[Proof of Proposition \ref{prop:uniqueness2}] 
Suppose that $f_0=L\fI+b$ and also $f_0=L'\fI+b'$, so we have to show $L=L'$ and $b=b'$. 
 By \ref{DR} and Lemma \ref{L:Vspan}, we can write any $v\in V_\III$ in the form $v=\lambda(\fI(x_1)-\fI(x_2))$, with $x_1,x_2\in X$ and $\lambda>0$. Applying $L$ or $L'$ we find $Lv=\lambda(f_0(x_1)-f_0(x_2))=L'v$;  therefore $L=L'$. Moreover, we have $b=f_0(x_1)-L\fI(x_1)=f_0(x_1)-L'\fI(x_1)=b'$.
 \end{proof}

\begin{proof}[Proof of Proposition \ref{prop:uniqueness3}] 
First we claim that $g_0$ and $g'_0$ are co-convex. Fix $x_1,x_2\in X$ and $\aa\in [0,1]$. Since $\fI$ and $g_0$ are co-convex, there exists $x_3\in X$ such that 
$\fI(x_3)=\aa \fI(x_1)+(1-\aa)\fI(x_2)$
and $g_0(x_3)=\aa g_0(x_1)+(1-\aa)g_0(x_2)$. Similarly there exists
$x'_3\in X$ such that 
$\fI(x'_3)=\aa \fI(x_1)+(1-\aa)\fI(x_2)$
and $g'_0(x'_3)=\aa g'_0(x_1)+(1-\aa)g'_0(x_2)$. 
By Theorem~\ref{T:infinitaryH-3}, \ref{P1} holds. 
Since $\fI(x'_3)=\fI(x_3)$,  we must, by \ref{P1}, have $g_0(x'_3)=g_0(x_3)$. 
Therefore $g_0(x'_3)=\aa g_0(x_1)+(1-\aa)g_0(x_2)$,
 establishing that $g_0$ and $g'_0$ are co-convex. 

Thus, by Lemma~\ref{L:coconvex0},
there is a unique linear order isomorphism $L \colon  V\to V'$ and 
unique $b_0\in V'$ such that $g'_0 = Lg_0 + b_0$. 

By assumption, for each $i\in\III$, 
$f_i$ and $g_i$ are co-convex, and (by \ref{DR}) $f_i$ is pervasive. 
As described before Lemma~\ref{L:coconvex0}, for any $x_0\in X$, $g_i-g_i(x_0)$ is a pervasive map $X\to \Span(g_i(X)-g_i(X))$. Since $f_i$ is also pervasive, Lemma~\ref{L:coconvex0} gives a unique linear order isomorphism $M_i\colon V_i\to\Span(g_i(X)-g_i(X))$ and unique $c_i\in V$ such that $g_i -g_i(x_0)= M_i f_i +c_i$. Thus $g_i=M_if_i+c_i+g_i(x_0)$. It is easy to verify from the uniqueness claim in Lemma~\ref{L:coconvex0} that both $M_i$ and $c_i+g_i(x_0)$ are uniquely determined by this equation, independently of $x_0$.  
As $i$ varies these together define $a\in \ProdI V$ and a linear order isomorphism $M\colon V_\III\to \ProdI \Span(g_i(X)-g_i(X))$ such that $\gI=M\fI+a$.
Now fix $i\in\III$ and $x\in X$. By \ref{DR} and Lemma \ref{L:Vspan}, there exist $y,z\in X$ and $\lambda>0$ such that $\lambda(\fI(y)-\fI(z))=1_if_i(x)$, where $1_i$ is the inclusion of $V_i$ into $V_\III$. 
 We therefore have 
$\lambda(\gI(y)-\gI(z))=\lambda(M\fI(y)-M\fI(z))=M1_if_i(x)=1_i M_i f_i(x) = 1_i(g_i(x)-a_i)$, where in the last two terms $1_i$ is the inclusion of $V$ into $\ProdI V$ as the $i$th factor.
Applying $S$ and using the fact that it extends summation we find
$g_0(y)-g_0(z)=S\gI(y)-S\gI(z)=(g_i(x) - a_i)/\lambda$.

A parallel argument gives
$g'_0(y)-g'_0(z)=(g'_i(x) - a'_i)/\lambda$ for some $a'_i\in V'$. 
Since $g'_0=Lg_0+b_0$, we have 
$g'_0(y)-g'_0(z)=Lg_0(y)-Lg_0(z)$; so in combination we find   
\begin{equation}
\label{eq:L} 
Lg_i(x)-La_i=g'_i(x)-a'_i.
\end{equation}
Rearranging, 
$g'_i=Lg_i+(a'_i-La_i)$. We set $b_i=(a'_i-La_i)\in V'$.
Note that by \eqref{eq:L}, $b_i=g'_i(x)-Lg_i(x)$ for every $x\in X$, and is uniquely determined by this equation. 

For the last statement, suppose first that $v\in\SumI V'$. We have
\[
LS((L\inv v_i)_{i\in\III})=L(\sum_{i\in\III} L\inv(v_i))=\sum_{i\in\III} LL\inv(v_i)=\sum_{i\in\III} v_i = S'(v),
\]
where every sum has finitely many non-zero terms.
On the other hand, suppose 
$v\in \Span(\gI'(X)-\gI'(X))$. 
Using Lemma \ref{L:Vspan} we can write $v$ in the form $v=\lambda(\gI'(x)-\gI'(y))$ with $x,y\in X$ and $\lambda>0$.  By equation \eqref{eq:L}, 
we have $L\inv v_i=\lambda(g_i(x_i)-g_i(y_i))$,
so $S((L\inv v_i)_{i\in\III})=\lambda(g_0(x)-g_0(y))$ and then $LS((L\inv v_i)_{i\in\III}) = \lambda(g'_0(x)-g'_0(y))=S'v$.
\end{proof}

\begin{proof}[Proof of Lemma~\ref{L:SI}]
Suppose that $\Geq$ is a preorder on $X$. It is straightforward to check that $\Geq$ satisfies strong independence if it has a mixture-preserving representation. 
Conversely, if $\Geq$ satisfies strong independence, 
let $V$ be the vector space $\Span X$ and $\iota \colon X \to V$ the inclusion. 
Taking $Z=X$ 
in Lemma~\ref{L:convex cone}, 
$\Geq$ and $\iota$ define a convex cone $C$ in $V$ and a linear preorder $\Geq_V$ 
such that $\iota$ then represents $\Geq$. 
Let $\ov{V}$ be the quotient of $V$ by the subspace $C\cap-C$, and let $L \colon V \to \ov{V}$ be the quotient map. 
Define a linear 
partial order
$\Geq_{\ov{V}}$ on $\ov{V}$ by $L(v) \Geq_{\ov{V}} L(w) \IFF v -w \in C$. 
It is clear that $L$ is a representation of $\Geq_V$, hence $L\iota$ is a representation of $\Geq$, 
with values in the partially ordered vector space $\ov{V}$.
To complete the proof of part \ref{part:SIa}, it suffices to apply the construction of pervasive representations described before Lemma~\ref{L:coconvex0}.

For part \ref{part:SIb}, fix any $x,y\in X$ and $\alpha\in[0,1]$. Let $z=\aa x + (1-\aa) y$. From the fact that the 
$f_i$  are mixture preserving, it follows that, for all $i$, $f_i(z)=\aa f_i(x)+(1-\aa)f_i(y)$, and therefore 
$\aa (f_i)_{i\in\II}(x)+(1-\aa)(f_i)_{i\in\II}(y)=(f_i)_{i\in\II}(z)$, establishing that  $(f_i)_{i\in \II}(X)$ is convex.
\end{proof}

\begin{proof}[Proof of Theorem~\ref{T:preferenceH}] 
By Lemma~\ref{L:SI}, we can choose co-convex, mixture-preserving representations $F_i\colon X\to V_i$, for $i\in\III\cup\{0\}$. We use these $F_i$ as the `$f_i$' in Theorem \ref{T:infinitaryH-3}, which then yields $V$, 
$g_i\colon X\to V$ and $S\colon \ProdI V\to V$.  For our current purposes we define $f_i\coloneqq g_i$.

The only thing left to show is that the $g_i$ are mixture-preserving.  
This follows from the construction of the $g_i$ in Theorem \ref{T:infinitaryH-3}, but here is a direct argument.
 Fix $x,y\in X$ and $\alpha\in[0,1]$.  By Theorem 
 \ref{T:infinitaryH-3}\ref{part:H-3c}, $g_i$ and 
 $F_i$ 
 are co-convex. So there exists $z\in X$ such that 
 $g_i(z)=\aa g_i(x)+(1-\aa)g_i(y)$ 
 and $F_i(z)=\aa F_i(x)+(1-\aa)F_i(y)$.  Since $F_i$ is mixture preserving, the second of these equations holds if and only if $F_i(z)=F_i(\aa x+(1-\aa)y)$, 
 or equivalently if and only if   $z\Eq_i \aa x+(1-\aa) y$. But then, since $g_i$ represents $\Geq_i$,  
 $g_i(z)=g_i(\aa x+(1-\aa)y)$, and 
 therefore $g_i(\aa x+(1-\aa)y)=\aa g_i(x)+(1-\aa)g_i(y)$, as desired. 
\end{proof}

\begin{proof}[Proof of Theorem~\ref{T:expectH}]
Lemma~\ref{L:expectational} 
 yields an expectational representation $f_i\colon X\to V_i$ of each $\Geq_i$; suppose $f_i(\mu)=\int_Y u_i\,\mathrm{d}\mu$. 
Let $\ov{V}_i=V_i/{\Eq_{V_i}}$ be the partially ordered quotient of $V_i$, and $\ov f_i\colon X\to \ov{V}_i$ be the composition of $f_i$ with the quotient map; $\ov f_i$ also represents $\Geq_i$. Since the $f_i$ are expectational, they are mixture preserving, and it follows from Lemma~\ref{L:SI} that they are co-convex; then the $\ov f_i$ are co-convex as well. By 
Theorem \ref{T:infinitaryH}\ref{part:infinitaryH P4} 
there is a partially ordered vector space $\ov V$ and a strictly positive linear map $\ov{L}\colon \ov{V}_\III\to \ov V$ such that $\ov{L}\ovfI$ represents $\Geq_0$, and such that the components $\ov{L}_i\colon \ov{V}_i\to \ov V$ are order embeddings.
Now define a preorder $\Geq_{V}$ on $V\coloneqq V_\III$ by $x\Geq_{V} y \iff \ov{L}\ov{x}\Geq_{\ov V} \ov{L}\ov{y}$, 
where $\ov{x},\ov{y}\in\ov{V}_\III$ are the images of $x,y\in V_\III$.  Let $L$ denote the set-theoretic identity map $V_\III\to V$. Since $\ov{L}$ is strictly positive, so is $L$. 
Since each $\ov{L}_i$ is an order-embedding, so is each $L_i$ 
(note $L$ is injective). 
Moreover, $L\fI$ represents $\Geq_0$. 
If $V'_i$ is the given separating vector space of linear functionals on $V_i$, then $V'\coloneqq \SumI V'_i$ is a separating vector space of linear functionals on $V_\III$, so on $V$; with respect to $V'$, we have $\fI(\mu)=\int_Y \uI \,\mathrm{d}\mu$.
Because $L$ is the set-theoretic identity map we have $L \fI(\mu)=L\int_Y \uI\,\mathrm{d}\mu=\int_Y L \uI\,\mathrm{d}\mu$.

We now appeal to the proof of the left-to-right direction of Theorem \ref{T:infinitaryH-3}: the construction (which works even though $V_i$ and $V$ are now merely preordered vector spaces) takes as input $\fI$ and $L$ and yields $g_i\colon X\to  V$ representing $\Geq_i$ (for $i\in \III\cup\{0\}$) and $S\colon \ProdI V\to V$ extending summation such that $g_0=Sg$. Moreover, by construction $g_i=L_if_i$ for $i\in\III$, and $g_0=L\fI$. 
We find    $g_i(\mu)=L_i \int_Y u_i\,\mathrm{d}\mu =\int_Y L_iu_i\,\mathrm{d}\mu$, 
and, as already stated, $g_0(\mu)=\int_Y L\uI\,\mathrm{d}\mu$. Defining $U_0=L\uI$ and $U_i=L_iu_i$ for $i\in\III$, it remains to show that $U_0=SU_\III$, i.e. that $L\uI=S(L_iu_i)$. And indeed $S$ is constructed so that $S\circ (L_i)=L$. 
\end{proof}

\begin{proof}[Proof of Theorem~\ref{T:pooling-c}]
Since Lyapunov's Theorem (Theorem \ref{T:Lyapunov}) implies that the $f_i$ are co-convex, Theorem \ref{T:infinitaryH} yields the vector space $V$ and linear map $L$ with components $L_i$. The only additional points to be checked are that $V$ can be taken to be finite-dimensional and that $\sum_{i\in\III} L_if_i$ is an admissible vector measure. 

On the first point, one can harmlessly replace $V$ by its subspace $L(V_\III)$, since the latter contains the image of every $L_i$; since, under the hypotheses of the theorem, $V_\III$ is finite-dimensional, so is $L(V_\III)$. (Alternatively, $V$ as constructed in the proof of Theorem \ref{T:infinitaryH} will already be finite dimensional.)

On the second point, given that the $f_i$ are admissible, it is straightforward to verify that their product $\fI\colon X\to V_\III$ is admissible, and so then is the linear transform $L\fI=\sum_{i\in\III} L_i f_i$ of $\fI$. 
\end{proof}

\begin{proof}[Proof of Lemma~\ref{L:SI2}]
We begin with part \ref{part:SI2-b}. Given $\ov f$, we have to check that its restriction $f$ to $X$ is additive. For $A,B$ disjoint in $X$ we have
\[
\begin{split}
\tfrac12 f(A\cup B)
=\tfrac12\ov f(\chi_A+\chi_B)
= \tfrac12\ov f(\chi_A+\chi_B)+\tfrac12\ov f(\chi_{\emptyset})
=\ov f(\tfrac12\chi_A+\tfrac12\chi_B)\\
=\tfrac12 \ov f(\chi_A)+\tfrac 12\ov f(\chi_B)
=\tfrac12(f(A)+f(B)).
\end{split}
\]
Hence $f(A\cup B)=f(A)+f(B)$.
Here we have used the fact that $X$ is embedded into $\ov X$ by $A\mapsto \chi_A$, the assumption that $\ov f(\chi_\emptyset)=0$, the mixture preservation property (twice), and again the embedding of $X$.

Conversely, suppose given a vector measure $f\colon X\to V$. 
We essentially define $\ov f \colon \ov X \to V$ by 
setting $\ov f(F)=\int_S F \,{\mathrm{d}f}$ for each $F \in \ov X$. 
Explicitly, suppose $F$ is constant on each cell of the partition $\{E_1,\ldots,E_n\}\subset X$ of $S$, taking value $p_j$ on $E_j$. Then set 
$\ov f(F)= \sum_{j=1}^n p_j f(E_j)$. 
This $\ov f$ is clearly mixture preserving, and the restriction of $\ov f$ to $X$ is $f$.

For part \ref{part:SI2-a}, suppose given $\ov \Geq$ satisfying strong independence; by Lemma~\ref{L:SI}, it admits a mixture-preserving representation $\ov f\colon \ov X\to V$. Subtracting a constant, we can assume $\ov f(\chi_\emptyset)=0$. Thus, by part \ref{part:SI2-b}, $\ov f$ restricts to a vector measure $f$ on $X$, which automatically represents $\Geq$. 

Conversely, if $\Geq$ can be represented by a vector measure $f\colon X\to V$, then, by part 
\ref{part:SI2-b}, 
$f$ extends to a mixture-preserving function $\ov f\colon \ov X\to V$; the preorder $\ov \Geq$ on $\ov X$ represented by $\ov f$ satisfies strong independence, by Lemma~\ref{L:SI}, and its restriction to $X$ is $\Geq$.
\end{proof}

\begin{proof}[Proof of Theorem~\ref{T:convexification}]
From Lemma~\ref{L:SI} we get co-convex mixture-preserving representations $\ov f_i\colon \ov X\to V_i$ of each $\ov\Geq_i$. 
Subtracting a constant in each case, we can assume $\ov f_i(\chi_\emptyset)=0$. 

Theorem \ref{T:infinitaryH}\ref{part:infinitaryH P4} yields a 
strictly positive linear map 
$L\colon V_\III\to V$, for some partially ordered vector space $V$, 
with every $L_i$ an order embedding, and
such that $L\ovfI$ represents $\Geq_0$.  By Lemma~\ref{L:SI2}\ref{part:SI2-b}, the $\ov f_i$ restrict to representations $f_i$
of the $\Geq_i$ by vector measures. 
We recover the statement of the theorem by redefining $V_0\coloneqq V$  and $\ov f_0\coloneqq L\ovfI$.
 \end{proof}

\bibliographystyle{plainnat}

\end{document}